\documentclass[acmsmall,nonacm]{acmart}

\setcopyright{none}
\renewcommand\footnotetextcopyrightpermission[1]{}
\makeatletter
\let\@authorsaddresses\@empty
\makeatother

\usepackage{cleveref}
\Crefname{algocf}{Algorithm}{Algorithms}
\usepackage{booktabs}
\usepackage{xspace}
\usepackage{enumitem}
\usepackage{microtype}
\usepackage{mathpartir}
\usepackage{todonotes}
\usepackage{booktabs}
\usepackage{ebproof}
\usepackage{caption}
\usepackage{subcaption}
\usepackage{xargs}
\usepackage[onelanguage,noend]{algorithm2e}
\usepackage{listings}

\lstdefinelanguage{zelus}
   {morekeywords={
	let,in,rec,where,end,if,then,else,do,done,run,%
	open,%
	fun,node,hybrid,%
	match,with,automaton,emit,%
	pre,when,whenot,fby,merge,on,clock,%
	or,and,not,as,mod,%
	unless,until,continue,reset,every,await,%
	init,der,type,last,%
  period,local,present,%
  sample, observe, factor, infer, proba,%
  val,%
  APF,%
  },
    sensitive=true,
    morecomment=[n]{(*}{*)},
    morestring=[b]",
    escapechar=\%,
    columns=fullflexible,
    keepspaces=true,
    basicstyle=\ttfamily,
    mathescape=true,
    }

\def\zl{\lstinline[basicstyle=\normalsize\ttfamily]}

\makeatletter
\DeclareFontEncoding{LS2}{}{\@noaccents}
\makeatother
\DeclareFontSubstitution{LS2}{stix}{m}{n}
\DeclareSymbolFont{largesymbolsstix}{LS2}{stixex}{m}{n}
\DeclareMathDelimiter{\lBrace}{\mathopen} {largesymbolsstix}{"E8}{largesymbolsstix}{"0E}
\DeclareMathDelimiter{\rBrace}{\mathclose}{largesymbolsstix}{"E9}{largesymbolsstix}{"0F}
\DeclareMathDelimiter{\rParen}{\mathclose}{largesymbolsstix}{"DF}{largesymbolsstix}{"03}
\DeclareMathDelimiter{\lParen}{\mathopen} {largesymbolsstix}{"DE}{largesymbolsstix}{"02}

\newtheorem{theorem}{Theorem}[section]
\newtheorem{lemma}{Lemma}[section]
\theoremstyle{definition}
\newtheorem*{example*}{Example}
\theoremstyle{remark}
\newtheorem*{remark*}{Remark}

\newenvironment{sketch}{\textsc{Proof Sketch.}}{\hfill$\square$}

\title{Density-Based Semantics for Reactive Probabilistic Programming}

\author{Guillaume Baudart}
\affiliation{
    \institution{ENS -- PSL University -- CNRS -- Inria}
    \country{France}}

\author{Louis Mandel}
\affiliation{
    \institution{IBM Research}
    \country{USA}}

\author{Christine Tasson}
\affiliation{
    \institution{ISAE Supaero}
    \country{France}}

\begin{abstract}
Synchronous languages are now a standard industry tool for critical embedded systems.
Designers write high-level specifications by composing streams of values using block diagrams.
These languages have been extended with Bayesian reasoning to program state-space models which compute a stream of distributions given a stream of observations.
However, the semantics of probabilistic models is only defined for scheduled equations -- a significant limitation compared to dataflow synchronous languages and block diagrams which do not require any ordering.

In this paper we propose two schedule agnostic semantics for a probabilistic synchronous language.
The key idea is to interpret probabilistic expressions as a stream of un-normalized density functions which maps random variable values to a result and positive score.
The co-iterative semantics interprets programs as state machines and equations are computed using a fixpoint operator.
The relational semantics directly manipulates streams and is thus a better fit to reason about program equivalence.
We use the relational semantics to prove the correctness of a program transformation required to run an optimized inference algorithm for state-space models with constant parameters.

 \end{abstract}

\begin{document}
\maketitle

\newenvironment{stack}{\begin{array}[t]{@{}l@{}}}{\end{array}}

\newcommand{\normal}[2]{\ensuremath{\mathcal{N}(#1, #2)}}
\newcommand{\typeof}[1]{\ensuremath{\mathit{typeOf}(#1)}}

\newcommand{\ttf}[1]{\ensuremath{{\texttt{#1}}}}
\newcommand{\mit}[1]{\ensuremath{{\ \mathit{#1} \ }}}
\newcommand{\is}{\ensuremath{\leftarrow}}
\newcommand{\dom}[1]{\ensuremath{\mathit{dom}(#1)}}
\newcommand{\im}[1]{\ensuremath{\mathit{im}(#1)}}
\newcommand{\kwf}[1]{\ensuremath{\mathtt{\ensuremath{\color{blue!50!black}{#1}}}}}
\newcommand{\op}{\ensuremath{\mathit{op}}}
\newcommand{\attr}[2]{\ensuremath{#1.\texttt{#2}}}
\newcommand{\None}{\ensuremath{\mathit{None}}}
\newcommand{\nil}{\ensuremath{\mathit{nil}}}
\newcommand{\Some}[1]{\ensuremath{\mathit{Some}\ #1}}
\newcommand{\RV}{\ensuremath{\mathrm{RV}}}
\newcommand{\CRV}{\ensuremath{\mathrm{CRV}}}

\newcommand{\fun}[1]{\ensuremath{\lambda #1. \;}}
\newcommand{\letin}[1]{\mathit{let} \; #1 \; \mathit{in} \;}
\newcommand{\fix}[1]{\ensuremath{\mathit{fix}(#1)}}
\newcommand{\norm}[1]{\ensuremath{\mathit{norm}(#1)}}

\newcommand{\ifthen}[3]{\ensuremath{
    \mathit{if} \; #1 \; 
    \begin{array}[t]{@{}l@{}l@{}}
        \mathit{then} \; &  \begin{array}[t]{@{}l} #2 \end{array}\\
        \mathit{else} \; &  \begin{array}[t]{@{}l} #3 \end{array}
    \end{array}}}

\newcommand{\optionmatch}[4]
{\ensuremath{
\begin{stack}
    \mathit{match} \ #1\  \mathit{with}\\ 
    \ \ \begin{array}[t]{@{}l@{\ }l@{\ }l@{}} 
        | \ \None &\to& #2\\
        | \ \Some{#3} &\to& #4
    \end{array}
\end{stack}}}

\newcommand{\sem}[1]{\left\llbracket #1 \right\rrbracket}
\newcommand{\psem}[1]{\left\lBrace #1 \right\rBrace}
\newcommand{\dsem}[1]{\left\lParen #1 \right\rParen}
\newcommand{\fsem}[1]{\left\lbag #1 \right\rbag}
\newcommand{\statetype}{\ensuremath{S}}
\newcommand{\valuetype}{\ensuremath{V}}

\newcommand{\zldef}[2]{\ensuremath{\kwf{let}\ #1\ \ttf{=}\ #2}}
\newcommand{\zlnode}[3]{\ensuremath{\kwf{node}\ #1\ #2\ \ttf{=}\ #3}}
\newcommand{\zlproba}[3]{\ensuremath{\kwf{proba}\ #1\ #2\ \ttf{=}\ #3}}
\newcommand{\zlprob}{\ensuremath{\ttf{prob}}}
\newcommand{\zlpair}[2]{\ensuremath{\ttf{(}#1\ttf{,}#2\ttf{)}}}
\newcommand{\zlop}[1]{\ensuremath{\op\ttf{(}#1\ttf{)}}}
\newcommand{\zlunit}[1]{\ensuremath{\ttf{()}}}
\newcommand{\zlapp}[2]{\ensuremath{#1\ttf{(}#2\ttf{)}}}
\newcommand{\zllast}[1]{\ensuremath{\kwf{last}\ #1}}
\newcommand{\zlwhere}[2]{\ensuremath{#1\ \kwf{where\ rec}\ #2}}
\newcommand{\zlletin}[3]{\ensuremath{\kwf{let} \ #1 \ \ttf{=} \ #2 \ \kwf{in} \ #3}}
\newcommand{\zlpresent}[3]{\ensuremath{\kwf{present}\ #1\ \to \ #2\ \kwf{else}\ #3}}
\newcommand{\zlreset}[2]{\ensuremath{\kwf{reset}\ #1\ \kwf{every}\ #2}}
\newcommand{\zlsample}[2]{\ensuremath{\kwf{sample}\ttf{(}#2\ttf{)}}}
\newcommand{\zlobserve}[2]{\ensuremath{\kwf{observe}\ttf{(}#1\ttf{,}#2\ttf{)}}}
\newcommand{\zlfactor}[1]{\ensuremath{\kwf{factor}\ttf{(}#1\ttf{)}}}
\newcommand{\zlinfer}[1]{\ensuremath{\kwf{infer}\ttf{(}#1\ttf{)}}}
\newcommand{\zleq}[2]{\ensuremath{#1\ \ttf{=}\ #2}}
\newcommand{\zlinit}[2]{\ensuremath{\kwf{init}\ #1\ \ttf{=}\ #2}}
\newcommand{\zland}[2]{\ensuremath{#1\ \kwf{and}\ #2}}
\newcommand{\zlisample}[3]{\zlinit{#1}{\zlsample{#2}{#3}}}

\newcommand{\zlapfinfer}[5]{\ensuremath{\kwf{APF.infer}\ttf{(}#3^{#1}\ttf{,} \ #4^{#2} \ttf{,}\ #5 \ttf{)}}}
\newcommand{\zlfmodel}[1]{\ensuremath{\attr{#1}{model}}}
\newcommand{\zlfprior}[1]{\ensuremath{\attr{#1}{prior}}}

\newcommand{\Stream}[1]{\ensuremath{#1^\omega}}

\newcommand{\const}[1]{\ensuremath{\mathit{const} \ #1}}
\newcommand{\hd}[1]{\ensuremath{\mathit{hd} \ #1}}
\newcommand{\tlOp}{\ensuremath{\mathit{tl}}}
\newcommand{\tl}[1]{\ensuremath{\tlOp \ #1}}
\newcommand{\mapOp}{\ensuremath{\mathit{map}}}
\newcommand{\map}[2]{\ensuremath{\mapOp \ #1 \ #2}}
\newcommand{\zip}[2]{\ensuremath{\mathit{zip} \ #1 \ #2}}
\newcommand{\mergeOp}{\ensuremath{\mathit{merge}}}
\newcommand{\merge}[3]{\ensuremath{\mergeOp \ #1 \ #2 \ #3}}
\newcommand{\whenOp}{\ensuremath{\mathit{when}}}
\newcommand{\when}[2]{\ensuremath{#1\ \whenOp \ #2}}
\newcommand{\slicerOp}{\ensuremath{\mathit{slicer}}}
\newcommand{\slicer}[2]{\ensuremath{\slicerOp\ #1 \ #2}}
\newcommand{\cumprod}[2]{\ensuremath{\mathit{cumprod} \ #1 \ #2}}
\newcommand{\pdf}[2]{\ensuremath{\mathit{pdf}_{#1}(#2)}}
\newcommand{\sample}[2]{\ensuremath{\mathit{icdf}_{#1}(#2)}}
\newcommand{\integ}[3]{\ensuremath{\mathit{integ}_{#1} \ #2 \ #3 }}

\newcommand{\dist}[1]{#1~\textrm{dist}}

\newcommandx{\cosem}[6][6= ]
{\ensuremath{\begin{array}{lll}
    \semi{#1}_{#2} &=& \begin{stack} #3 \end{stack}\\
    \sems{#1}_{#2}(#4) &=& #6 \begin{stack} #5 \end{stack}
\end{array}}}

\newcommand{\pcosem}[5]
{\ensuremath{\begin{array}{lll}
    \psemi{#1}_{#2} &=& \begin{stack} #3 \end{stack}\\
    \psems{#1}_{#2}(#4) &=& \begin{stack} #5 \end{stack}
\end{array}}}

\newcommand{\semi}[1]{\sem{#1}^{\textrm{init}}}
\newcommand{\sems}[1]{\sem{#1}^{\textrm{step}}}
\newcommand{\semdi}[1]{\sem{#1}^{\textsc{d}\,\textrm{init}}}
\newcommand{\semds}[1]{\sem{#1}^{\textsc{d}\,\textrm{step}}}

\newcommand{\psemi}[1]{\psem{#1}^{\textrm{init}}}
\newcommand{\psems}[1]{\psem{#1}^{\textrm{step}}}

\newcommand{\dsemi}[1]{\dsem{#1}^{\textrm{init}}}
\newcommand{\dsems}[1]{\dsem{#1}^{\textrm{step}}}
\newcommand{\dssems}[2]{\dsem{#1}_{#2}^{\textsc{s}\,\textrm{step}}}

\newcommand{\fsems}[1]{\fsem{#1}^{\textrm{step}}}

\newcommand{\rsemexp}[3]{#1 \vdash #2 \downarrow #3}
\newcommand{\rsemeq}[2]{#1 \vdash #2}
\newcommand{\prsemexp}[3]{#1 \vdash #2 \Downarrow #3}
\newcommand{\prsemeq}[3]{#1 \vdash #2 : #3}

\newcommand{\vdashc}[3]{\ensuremath{{#1} \vdash {#2} : {#3}}}
\newcommand{\vdashceq}[4]{\ensuremath{{#1}, {#2} \vdash {#3} : {#4}}}
\newcommand{\vdashl}[2]{\ensuremath{{#1} \vdash^c {#2}}}
\newcommand{\vdashleq}[3]{\ensuremath{{#1} \vdash^c {#2} : {#3}}}
\newcommand{\notvdashl}[2]{\ensuremath{{#1} \not\vdash^c {#2}}}

\newcommand{\apfenv}{\ensuremath{\Phi}}
\newcommand{\apftype}{\ensuremath{\phi}}

\newcommand{\compile}[2]{\mathcal{C}_{#1}(#2)}
\newcommand{\flatten}[1]{\mathcal{F}(#1)}

\section{Introduction}

Synchronous programming languages~\cite{synchronous-twelve-years-later} were introduced for the design of critical embedded systems.
In dataflow languages such as Lustre~\cite{lustre:ieee91}, system designers write high-level specifications by composing infinite streams of values, called \emph{flows}.
Flows progress \emph{synchronously}, paced on a global logical clock.
The expressiveness of synchronous languages is deliberately restricted.
Specialized compilers can thus generate efficient and correct-by-construction embedded code with strong guarantees on execution time and memory consumption.
This approach was inspired by block diagrams, a popular notation to describe control systems.
Built on these ideas, Scade~\cite{lucy:tase17} is now a standard tool in automotive and avionic industries.

Probabilistic languages~\cite{BinghamCJOPKSSH19,goodman_stuhlmuller_2014,pmlr-v84-ge18b,gen} extend general purpose programming languages with probabilistic constructs for Bayesian inference.
Following a Bayesian approach, a program describes a probability distribution, the \emph{posterior} distribution, using initial beliefs on random variables, the \emph{prior} distributions, that are conditioned on observations.

At the intersection between these two lines of research, ProbZelus~\cite{rppl-short} is a probabilistic extension of the synchronous dataflow language Zelus~\cite{lucy:hscc13}.
ProbZelus combines, in a single source program, deterministic controllers and probabilistic models that can interact with each other to perform \emph{inference-in-the-loop}.
A classic example is the \emph{Simultaneous Localization and Mapping problem}~(SLAM)~\cite{Montemerlo02-fastslam} where an autonomous agent tries to infer both its position and a map of its environment to adapt its trajectory.

The probabilistic model of the SLAM involves two kinds of parameters.
The position is a \emph{state parameter} represented by a stream of random variables. 
At each instant, a new position must be estimated from the previous position and the observations.
The map is a \emph{constant parameter} represented by a random variable whose value is progressively refined from the \emph{prior} distribution with each new observation.
This type of problem mixing constant parameters and state parameters are instances of \emph{State-Space Models}~(SSM)~\cite{chopin2020SMC}.
Any ProbZelus program can be expressed as a SSM.

\paragraph{Probabilistic semantics and scheduling}
ProbZelus semantics~\cite{rppl-short} is defined in a co-iterative framework where expressions are interpreted as state machines.
Following~\cite{kozen81,staton17}, a probabilistic expression computes a stream of \emph{measures}.
The semantics of an expression with a set of local declarations integrates the semantics of the main expression over all possible values of the local variables.
Unfortunately, this semantics yields nested integrals that are only well defined if the declarations are \emph{scheduled}, i.e., ordered according to data dependencies.

This is a significant limitation compared to synchronous dataflow languages where sets of mutually recursive equations can be written in any order.
Besides, the compiler implements a series of source-to-source transformations which often introduces new variables in arbitrary order.
Scheduling local declarations is one of the very last compilation passes~\cite{lucy:hscc13}. 
The semantics of ProbZelus is thus far from what is exposed to the programmer, and more importantly, prevents reasoning about most program transformations and compilation passes.

In this paper, we show how to extend the schedule agnostic semantics of dataflow synchronous languages~\cite{pouzet-cmcs98, BourkeBDLPR17} for probabilistic programming.
The key idea is to interpret a probabilistic expression as a stream of un-normalized density functions which map random variable values to a result and a positive score.

\paragraph{Contributions}
In this paper, we present the following contributions:
\begin{itemize}
\item We introduce in \Cref{sec:coit} a new density-based co-iterative semantics and show that sets of mutually recursive equations in arbitrary order can be interpreted using a fixpoint operator.
We prove that this semantics is equivalent to the original ProbZelus semantics.
\item We introduce in \Cref{sec:rel} an alternative relational semantics which abstracts away the state machines and directly manipulates streams which simplifies reasoning about program equivalence.
We prove that this semantics is equivalent to the co-iterative semantics.
\item We define in \Cref{sec:apf} a program transformation required to run an optimized inference algorithm for state-space models with constant parameters.
We use the relational semantics to prove the correctness of the transformation. 
\end{itemize}

\section{Example}
\label{sec:example}

To motivate our approach, consider the ProbZelus model of \Cref{fig:tracker} adapted from~\cite{chopin2020SMC}[Section~2.4.1].
The goal is to estimate at each instant the position of a moving boat given noisy observations from a marine radar.
A rotating antenna sweeps a beam of microwaves and detects the boat when the beam is reflected back to the antenna.
The radar then estimates the position from noisy measurements of its angle and the echo delay.

\begin{figure}
\begin{subfigure}[b]{0.55\textwidth}
\begin{lstlisting}[basicstyle=\small\ttfamily,xleftmargin=0pt, numbers=left,numberblanklines=false]
proba tracker(y_obs) = x where            $\label{ex:line:trackerdef}$
  rec init x = x_init                                   $\label{ex:line:init}$
  and x = sample(gaussian(f(last x), sx))               $\label{ex:line:f}$
  and y = g(x)                                          $\label{ex:line:g}$
  and () = observe(gaussian(y, sy), y_obs)              $\label{ex:line:obs}$

node main(y_obs) = msg where
  rec x_dist = infer (tracker (y_obs))  $\label{ex:line:infer}$
  and msg = controller(x_dist)            $\label{ex:line:controller}$
\end{lstlisting}
\end{subfigure}
\begin{subfigure}[b]{0.37\textwidth}
\includegraphics[scale=0.35]{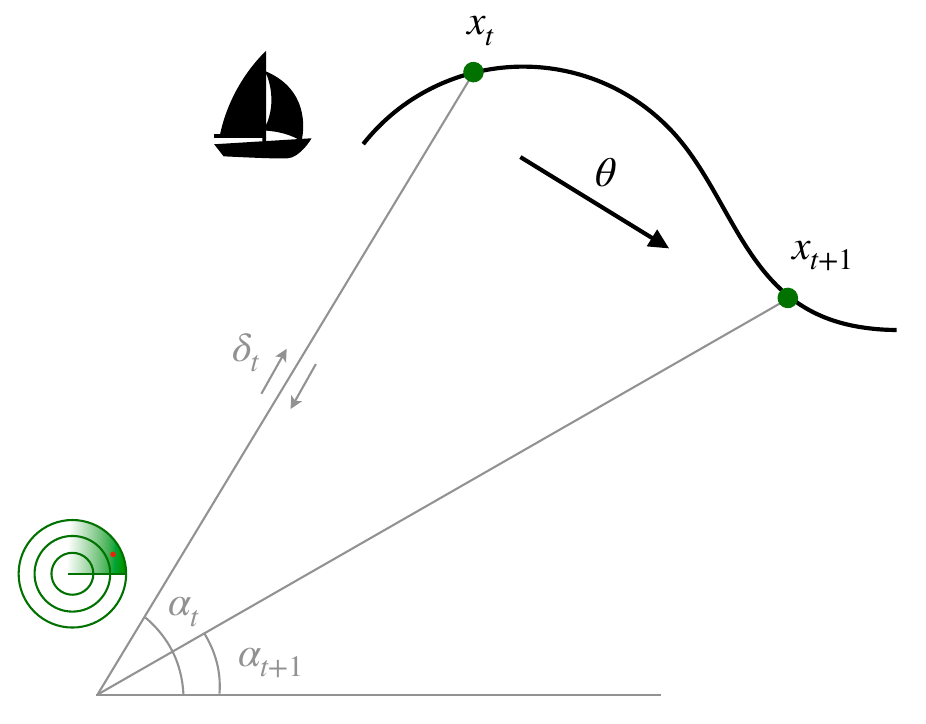}
\end{subfigure}
\caption{Tracking a moving boat with a marine radar in ProbZelus.
$x_t$ is the position of the boat in Cartesian coordinates.
We assume a linear motion model $f(x_{t-1}) = x_{t-1} + \theta$, and $g(x_t) = (\alpha_t, \delta_t)$ returns the radar angle $\alpha_t = \mathrm{atan}(x_t[1]/x_t[0])$, and the echo delay $\delta_t = 2 * \lVert x_t \rVert / c$ where $c$ is speed of light.}
\label{fig:tracker}
\end{figure}

The keyword \zl{proba} indicates the definition of a probabilistic stream function.
Line~\ref{ex:line:trackerdef}, the model \zl{tracker} takes as input a stream of observations \zl{y_obs} and returns a stream of positions \zl{x}.
The model uses a transition function \zl{f} to estimate the current position (e.g., using a linear motion model), and a projection function \zl{g} to compute the observable quantities from the state (e.g., angle and echo delay).
Line~\ref{ex:line:f} uses the \zl{sample} operator to specify that \zl{x} is Gaussian distributed around \zl{f(last x)}. \zl{last x} refers to the previous position of the boat initialized Line~\ref{ex:line:init} with the \zl{init} keyword.
Line~\ref{ex:line:obs} uses the \zl{observe} operators to condition the model assuming that the observations \zl{y_obs} are Gaussian distributed around~\zl{y = g(x)}.
The initial position \zl{x_init} and the noise parameters \zl{sx} and \zl{sy} are global constants.

\subsection{Kernel-based co-iterative semantics}
ProbZelus original semantics~\cite{rppl-short} is a co-iterative semantics where expressions are interpreted as state machines characterized by an initial state and a transition function.
Given the current state, the transition function of deterministic expression returns the next state and a value.
Following~\cite{staton17}, the transition function of a probabilistic expression returns a \emph{measure} over all possible pairs (next state, value).

For instance, if \zl{f} and \zl{g} are deterministic and stateless, the transition function of \zl{tracker} is the following (omitting empty states for stateless expressions) where~$\mathcal{N}$ is the normal distribution, $\delta$~the Dirac delta distribution, and $\mathit{pdf}_{d}$ the probability density function of~$d$.
 \begin{small}
\begin{align}
\psems{\ttf{tracker(y\_obs)}}_\gamma(p_x) 
&= \int \normal{f(p_x)}{s_x}(dx) 
   \int \delta_{g(x)}(dy) \ 
   \pdf{\normal{y}{s_x}}{y_{\rm obs}}
   * \delta_{x, x} \nonumber\\
&= \int
    \normal{f(p_x)}{s_x}(dx)\ 
    \pdf{\normal{g(x)}{s_x}}{y_{\rm obs}}
    * \delta_{x, x} \label{eq:tracker:sch}
\end{align}
\end{small}

\noindent
Given the current state $p_x$ (the previous value of \zl{x}), 
the transition function integrates over all possible values for \zl{x}, then all possible values for \zl{y}, weights each execution by the likelihood of the observation (i.e., the value of the density function on \zl{y_obs}), and returns a measure over the new state $x$ (that will be used as the previous position in the next step), and the results $x$.

\paragraph{Inference}
In ProbZelus, the explicit \zl{infer} operator computes the stream of distributions described by a model. 
The state of \zl{infer} is a measure over current states.
The transition function integrates the semantics of the model over all possible states, normalizes the resulting measure to obtain a distribution, and splits it into a distribution of next states and a distribution of values.
The runtime iterates this process from a distribution of initial states to compute a stream of distributions.

In ProbZelus, inference can be composed with deterministic control to perform \emph{inference-in-the-loop}.
For instance, Line~\ref{ex:line:controller}, the distribution \zl{x_dist} is exploited by a controller, e.g., to send a message to guide the boat.

\paragraph{Scheduling}
Local declarations such as \zl{x} and \zl{y} yield nested integrals in \Cref{eq:tracker:sch} that are only well-defined if these definitions are ordered according to data dependencies.
The ProbZelus semantics defined in~\cite{rppl-short} thus focuses on a  kernel language where local declarations are all \emph{scheduled}.
But imposing a valid schedule is a significant limitation compared to synchronous dataflow languages which manipulate mutually recursive equations in arbitrary order.

\subsection{Density-based co-iterative semantics}
In this paper we first propose a new density-based co-iterative semantics for ProbZelus inspired by the semantics of the popular probabilistic language Stan~\cite{GorinovaGS19} where a program defines an un-normalized density function over the random variables.
Instead of manipulating measures via integration, probabilistic expressions are now similar to deterministic expressions, but the transition function now takes one additional argument --- a random seed for all random variables ---  and returns one additional output --- a positive score, or weight, which measures the quality of the output w.r.t. the model.

\medskip
On the example of \Cref{fig:tracker} we have:
\begin{small}
\begin{align}
\dsems{\ttf{tracker(y\_obs)}}_\gamma(p_x, r) \  = \ 
\begin{stack}
\letin{\mu_x = \normal{f(p_x)}{s_x}}\\
\letin{x = \sample{\mu_x}{r}}\\
\letin{\mu_y = \normal{g(x)}{s_y}}\\
x, x, \pdf{\mu_y}{y_{\rm{obs}}}
\end{stack} \label{eq:tracker:dens}
\end{align}
\end{small}

\noindent
The additional argument $r$ corresponds to the random seed for the \zl{sample} operator, i.e., an element of $[0, 1]$ that is mapped to a sample of a distribution $d$ using inverse transform sampling~\cite{devroye2006nonuniform}.
The model computes the sample $x$ associated to the random seed ($\mathit{icdf}_d$ is the inverse of the cumulative distribution function of $d$), and returns the new state $x$, the result~$x$, and a weight which captures the likelihood of the observation (the density of the distribution $\normal{g(x)}{s_y}$ at $y_{\rm obs}$).

\paragraph{Inference}
At each step, the \zl{infer} operator first computes the un-normalized measure which associates each pair (state, result) to its weight, i.e., for a model~$e$, if $\dsem{e}_\gamma(m, r) = m', v, w$, $\zlinfer{e}$ computes the measure $\int_{[0, 1]^p} w * \delta_{m', v}\ dr$ where $p$ is the number of random variables in $e$.
On the example of \Cref{fig:tracker}, we can check that this measure corresponds to the original semantics of \Cref{eq:tracker:sch}.
\begin{small}
\begin{align*}
\int_{[0, 1]} \letin{x, x, w = \dsems{\ttf{tracker(y\_obs)}}_\gamma(p_x, r)} w * \delta_{x, x} \ dr \qquad \qquad\\
\begin{aligned}
& = \int_{[0, 1]} \letin{x = \sample{\normal{f(p_x)}{s_x}}{r}}
    \ \pdf{\normal{g(x)}{s_y}}{y_{\rm obs}}
    * \delta_{x, x}\ dr \\
& = \int \normal{f(p_x)}{s_x}(dx) \ \pdf{\normal{g(x)}{s_y}}{y_{\rm obs}}
    * \delta_{x, x}
\end{aligned}
\end{align*}
\end{small}

\noindent
The semantics of \zl{infer} is then similar to its interpretation in the original kernel-based semantics, i.e., 1)~integrate over all possible state, 2)~normalize the measure, 3)~split the result into a distribution of next states and a distribution of values.
We prove in \Cref{sec:coit:inf} that this semantics is equivalent to the kernel-based semantics, i.e., the \zl{infer} operator yields the same stream of distributions.

\paragraph{Mutually recursive equations}
The original co-iterative semantics for dataflow synchronous languages~\cite{pouzet-cmcs98} interprets mutually recursive equations in arbitrary order with a fixpoint operator in a flat complete partial order~(CPO) where variables are either undefined or set to a value.
In the density-based semantics, the transition functions of probabilistic equations are similar to their deterministic counterparts with additional inputs/outputs.
Compared to the kernel-based semantics, there are no longer nested integrals and a fixpoint operator can be defined to interpret sets of probabilistic equations.

Consider a variant of the example of \Cref{fig:tracker} where we swap Line~\ref{ex:line:f} and Line~\ref{ex:line:g}.
For a state $p_x$ and a random seed $r$, the semantics of the local declarations in \zl{tracker} is the fixpoint of the following function $F$ starting from the least element $[x \is \bot, y \is \bot]$.

\bigskip
\begin{small}
\centering
\begin{tabular}{p{0.49\textwidth}p{0.49\textwidth}}
\(
F(\rho) \  = \ 
\begin{stack}
\letin{y = g (\rho(\ttf{x}))}\\
\letin{\mu_x = \normal{f(p_x)}{s_x}}\\
\letin{x = \sample{\mu_x}{r}}\\
{[\ttf{x} \is x, \ttf{y} \is y]}
\end{stack}
\)
&
\(
\begin{array}[t]{lcl}
\rho_0 &=& [x \is \bot, y \is \bot]\\
\rho_1 &=& [x \is \sample{\mu_x}{r}, y \is \bot]\\
\rho_2 &=& [x \is \sample{\mu_x}{r}, y \is g(\sample{\mu_x}{r})]\\
\rho_3 &=& [x \is \sample{\mu_x}{r}, y \is g(\sample{\mu_x}{r})]
\end{array}
\)
\end{tabular}
\end{small}

\bigskip

\noindent
The fixpoint converges after~$3$ iterations.
Using the resulting environment, the semantics of \zl{tracker} then computes the next state, the result, and the weight which, after simplification, yields the same results as \Cref{eq:tracker:dens}.

\begin{small}
\begin{align}
\dsems{\ttf{tracker(y\_obs)}}_\gamma(p_x, r) \  = \ 
\begin{stack}
\letin{\rho = [x \is \sample{\mu_x}{r}, y \is g(\sample{\mu_x}{r})]}\\
\letin{x = \rho(\ttf{x})}\\
\letin{\mu_y = \normal{\rho(\ttf{y})}{s_y}}\\
\rho(\ttf{x}), \rho(\ttf{x}), \pdf{\mu_y}{y_{\rm{obs}}}
\end{stack}
\end{align}
\end{small}

\paragraph{Program equivalence}
Since deterministic expressions are interpreted as state machines, to prove program equivalence one must exhibit a bisimulation~\cite{Park1981}, i.e., a relation between the states of the two state machines.
Two deterministic expressions are equivalent if there exists a relation such that 1) the initial states are in relation, and 2) given two states in relation the transition function produces new states in relation and the same output.
The proof is done by unfolding the definition of the transition function.

Two probabilistic expressions are equivalent if they describe the same stream of measures.
Unfortunately, compared to deterministic expressions, the type of probabilistic expressions is asymmetric: given a state, the transition function describes a measure over pairs (next state, value).
The corresponding stream of measures is obtained by integrating at each step the transition function over all possible states computed at the previous step.
The bisimulation must thus relate measures of states through successive integrations, which is significantly harder to define and to check than in the deterministic case.

\subsection{Density-based relational semantics}
An alternative is to directly manipulate streams of values.
This is the approach used in the Vélus project\footnote{\url{https://velus.inria.fr}} to prove an end-to-end compiler for the dataflow synchronous language Lustre~\cite{BourkeBDLPR17, BourkeBP20, BourkeJPP21}.
The semantics of a stream function is defined as a relation between input streams and output streams.
In Vélus, most of the compilation passes are proven correct using this \emph{relational semantics}.
The translation to state machines is one of the very last passes and focuses on a normalized, scheduled subset of the language.

In this paper, we extend this relational semantics to probabilistic streams.
The key idea is to lift the density-based semantics to streams.
Given an environment $H$ mapping variable names to streams of values, and an array $R$ of random streams, the semantics of an expression returns a stream of pairs (value, weight): $H, R \vdash e \Downarrow (v, w)$.

In the example of \Cref{fig:tracker}, $R$ is a single stream of independent random elements $R_0 \cdot R_1 \cdot R_2 \cdot ...$ in $[0, 1]$, we can interpret \zl{tracker} in an environment that contains the observations \zl{y_obs}:

\begin{small}
\[
\begin{stack}
\begin{array}[t]{@{}rc@{~}c@{~}c@{~}c@{~}c@{~}c@{~}c@{~}c@{~}c@{~}c@{~~}c}
{[\ttf{y\_obs} \is y_{\rm obs}]}, R\ \vdash\ \ttf{tracker(y\_obs)}
\!\!\!\!\!\!\!\!&&\Downarrow& (x_0, w_0) &\cdot&
    (x_1, w_1) &\cdot&
    (x_2, w_2) &\cdot& ...\\[0.7em]
\mit{where} &\mu_x &=& 
        \normal{f(x_{\rm init})}{s_x} &\cdot& 
        \normal{f(x_0)}{s_x} &\cdot& 
        \normal{f(x_1)}{s_x} &\cdot&...\\
&x &=& 
        \sample{{\mu_x}_0}{R_0} &\cdot& 
        \sample{{\mu_x}_1}{R_1} &\cdot& 
        \sample{{\mu_x}_2}{R_2} &\cdot&...\\
&y &=&
        g(x_0) &\cdot& 
        g(x_1) &\cdot& 
        g(x_2) &\cdot& ...\\
&\mu_y &= &
        \normal{y_0}{s_x} &\cdot& 
        \normal{y_1}{s_x} &\cdot&
        \normal{y_2}{s_x} &\cdot& ...\\
&w &=&
    \pdf{{\mu_y}_0}{{y_{\rm{obs}}}_0} &\cdot&
    \pdf{{\mu_y}_1}{{y_{\rm{obs}}}_1} &\cdot&
    \pdf{{\mu_y}_2}{{y_{\rm{obs}}}_2} &\cdot& ... \\
\end{array}
\end{stack}
\]
\end{small}

\noindent
The semantics now directly manipulates streams.
At each step, the result is similar to the expression in \Cref{eq:tracker:dens}, but states are abstracted away. 
The result is a stream of pairs (value, weight).

\paragraph{Inference}
The semantics of \zl{infer} now operates on a stream of pairs (value, weight): $(v_0, w_0) \cdot (v_1, w_1) \cdot (v_2, w_2) \cdot ...$.
The \zl{infer} operator 1)~associates to each value $v_k$ the total weight of its prefix using a cumulative product $\overline{w_k} = \Pi_{i =0}^k w_i$, 2)~computes the un-normalized measure which associates each pair (state, result) to its total weight, and 3)~normalizes it to obtain a distribution of values.
The key difference with the density-based co-iterative semantics is that the integral is now over the infinite domain of streams.
We prove in \Cref{sec:sem:rel:inf} that this semantics is equivalent to the co-iterative density-based semantics, i.e., the \zl{infer} operator yields the same stream of distributions.

\paragraph{Mutually recursive equations}
Given the random streams $R$, the semantics of a set of probabilistic equations $H, R \vdash E : W$ checks that an environment mapping variable names to stream of values $H$ is compatible with all the equations in $E$, and that the combined weight of all sub-expressions is the stream $W$.
Since variables in an environment are not ordered, there is nothing special to do to interpret mutually recursive equations.
By construction the order of equations does not matter which greatly simplifies reasoning about compilation passes that introduce new equations in arbitrary order.
Of course, compared to the co-iterative semantics, the relational semantics is not executable since equations are only checked a posteriori for a given environment.

\paragraph{Program equivalence}
In the relational semantics, states are abstracted away and a probabilistic expression computes a stream of pairs (value, weight) where each element only depends on the random streams.
Two probabilistic expressions are equivalent if they describe the same stream of measures obtained by integrating at each step the result of the relational semantics over all possible random streams.
Since, random streams are uniformly distributed, if we can map the random streams of one expression to the random streams of the other, program equivalence can be reduced to the comparison of the streams of pairs (value, weight) computed by each expression.

\section{Background}
\label{background}

In this section we briefly summarize the key elements of the co-iterative semantics of ProbZelus.
Importantly, this semantics is only defined if all equations are ordered according to data dependencies.
We then recall the original co-iterative semantics of synchronous dataflow languages where sets of mutually recursive equations in arbitrary order are interpreted using a fixpoint operator.

\subsection{Syntax}

\begin{figure}
\[
\begin{array}[t]{@{}lrl@{}}
d &::=& \zldef{x}{e}\mid \zlnode{f}{x}{e}\mid \zlproba{f}{x}{e}\mid d\ d
\\[.5em]
e &::=& c \mid  x\mid \zlpair{e}{e}\mid \zlop{e} \mid \zllast{x}
	\mid \zlapp{f}{e}\mid \zlwhere{e}{E} \\
	&\mid& \zlpresent{e}{e}{e}\mid \zlreset{e}{e}\\
    &\mid& \zlsample{}{e}\mid \zlfactor{e}\mid \zlinfer{e}
\\[.5em]
E &::=& \zleq{x}{e} \mid \zlinit{x}{e} \mid \zland{E}{E}
\end{array}
\]
\caption{ProbZelus Syntax.}
\label{fig:syntax}
\end{figure}

The syntax of ProbZelus is presented in \Cref{fig:syntax}.
A program is a series of declarations~$d$.
A declaration can be a global variable \zl{let}, a deterministic stream function \zl{node}, or a probabilistic model \zl{proba}.
Each declaration has a unique name.
An expression can be a constant~$c$, a variable~$x$, a pair, an operator application $\zlop{e}$, the previous value of a variable $\zllast{x}$, a function call $\zlapp{f}{e}$, a local declaration $\zlwhere{e}{E}$ where $E$ is a set of mutually recursive equations, a lazy conditional $\zlpresent{c}{e_1}{e_2}$, or a reset construct $\zlreset{e_1}{e_2}$.
An equation is either a simple definition $\zleq{x}{e}$, an initialization $\zlinit{x}{e}$ (the delay operator $\zllast{x}$ can only be used on initialized variables), or a set of equations $\zland{E_1}{E_2}$.
In a set of equations, every initialized variable must be defined by another equation.

We add the classic probabilistic constructs to the set of expressions: $\zlsample{}{d}$ creates a random variable with distribution $d$, $\zlfactor{s}$ increments the log-density of the model, and $\zlinfer{m}$ computes the posterior distribution of a model $m$.
If $d$ is a distribution with a density function, we use the syntactic shortcut $\zlobserve{d}{x}$ for $\zlfactor{\pdf{d}{x}}$ which conditions the model on the assumption that $x$ was sampled from $d$.  

\subsection{Co-iterative semantics}
\label{sec:back:co-it}

The semantics of ProbZelus presented in \cite{rppl-short} extends the co-iterative semantics of dataflow synchronous languages~\cite{pouzet-cmcs98, emsoft23b}.
A type system statically identifies deterministic and probabilistic expressions~\cite[Section~3.2]{rppl-short} which have different interpretations.

In an environment $\gamma$ mapping variable names to values, a deterministic expression $e$ is interpreted as a state machine characterized by an initial state $\semi{e}_\gamma$ of type $\statetype$ and a transition function $\sems{e}_\gamma$ of type $\statetype \to \valuetype \times \statetype$ which given the current state returns a value and the next state.
A stream of values is then obtained by iteratively applying the transition function from the initial state.

\newcommand{\mask}[1]{\textcolor{black!40!white}{#1}}

\begin{small}
\[
\mask{(\semi{e}_{\gamma_0} =}\  m_0 \mask{)}\ 
\mask{\xrightarrow{\sems{e}_{\gamma_1}}} 
    \begin{array}[t]{c} m_1\\v_1 \end{array}
\mask{\xrightarrow{\sems{e}_{\gamma_2}}}
    \begin{array}[t]{c} m_2\\v_2 \end{array}
\mask{\xrightarrow{\sems{e}_{\gamma_3}}}
    \begin{array}[t]{c} m_3\\v_3 \end{array}
\mask{\to} \ ...
\]
\end{small}

\begin{figure}
\begin{small}
\[
\begin{array}{@{}lll@{}}

\psemi{e}_{\gamma} &=& \semi{e}_\gamma\\
\psems{e}_{\gamma}(m) &=& \letin{m', v = \sems{e}_\gamma(m)} \delta_{m',v}
\quad \text{if $e$ is deterministic}
\\\\
\psemi{\zlsample{}{e}}_{\gamma} &=& \semi{e}_\gamma\\
\psems{\zlsample{}{e}}_{\gamma}(m) &=&
    \letin{m', \mu = \sems{e}_\gamma(m)} \displaystyle \int \mu(dv)\ \delta_{m', v}
\\\\
\psemi{\zlfactor{e}}_{\gamma} &=& \semi{e}_\gamma\\
\psems{\zlfactor{e}}_{\gamma}(m) &=&
    \letin{m', v = \sems{e}_\gamma(m)} v * \delta_{m', ()}
\\\\
\psemi{
    \begin{array}[c]{@{}l@{}}
        e\ \begin{array}[t]{@{}l@{}} \kwf{where}\ 
        \begin{array}[t]{@{}l@{}} 
        \kwf{rec}\ \zlinit{x}{c}\\
        \kwf{and}\ \zleq{x}{e_x}\\
        \kwf{and}\ \zleq{y}{e_y}
        \end{array}
        \end{array}
    \end{array}
}_{\gamma} &=&
    c, \left(\psemi{e}_\gamma, \psemi{e_x}_\gamma, \psemi{e_y}_\gamma\right)
\\[2em]
\psems{
    \begin{array}[c]{@{}l@{}}
        e\ \begin{array}[t]{@{}l@{}} \kwf{where}\ 
        \begin{array}[t]{@{}l@{}} 
        \kwf{rec}\ \zlinit{x}{c}\\
        \kwf{and}\ \zleq{x}{e_x}\\
        \kwf{and}\ \zleq{y}{e_y}
        \end{array}
        \end{array}
    \end{array}
}_{\gamma}(p_x, (m, m_x, m_y)) &=&
    \begin{stack}
        \displaystyle \int \psems{e_x}_{\gamma + [\attr{x}{last} \is p_x]}(m_x)(d m_x', d v_x)\\[0.45em]
        \quad \displaystyle \int \psems{e_y}_{\gamma + [\attr{x}{last} \is p_x, x \is v_x]}(m_y)(d m_y', d v_y)\\[0.45em]
        \qquad \displaystyle \int \psems{e}_{\gamma + [\attr{x}{last} \is p_x, x \is v_x, y \is v_y]}(m)(d m', d_v)\\
        \qquad \quad \delta_{(v_x, (m', m_x', m_y')), v}
    \end{stack} 
\\\\
\semi{\zlinfer{e}}_{\gamma} &=& \semi{e}_\gamma\\
\sems{\zlinfer{e}}_{\gamma}(\sigma) &=&
    \begin{stack}
        \letin{\nu = \displaystyle \int \sigma(dm) \psem{e}_{\gamma}(m)}\\
        \letin{\overline{\nu} = \nu / \nu(\top)}\\
        (\pi_{1*}(\overline{\nu}), \pi_{2*}(\overline{\nu}))
    \end{stack}
\end{array}
\]
\end{small}
\caption{A simplified excerpt of the original ProbZelus co-iterative probabilistic semantics~\cite{rppl-short}.}
\label{fig:sem:coit:kernel}
\end{figure}

Following~\cite{staton17}, the semantics of a probabilistic expression is a state machine which computes a stream of \emph{kernels}.
Given the current state, the transition function $\psems{e}_\gamma$ of type $\statetype \to \Sigma_{\statetype \times \valuetype} \to [0, \infty)$ returns a measure over pairs (next state, value),\footnote{$\Sigma_{A}$ denotes the Borel $\sigma$-algebra over values of type $A$.} i.e., a function mapping measurable sets of pairs (next state, value) to a positive score.

\Cref{fig:sem:coit:kernel} shows a simplified excerpt of the semantics of probabilistic expressions from~\cite{rppl-short}.
In a probabilistic context, a deterministic expression is interpreted as the Dirac delta measure on the pair (state, value) returned by the deterministic semantics.\footnote{$\delta_x(U) = 1$  if $x \in U$ and $0$ otherwise.}
\zl{sample} evaluates its argument into a new state $m'$ and a distribution of values $\mu$, and returns a measure over pairs (new state, value).
\zl{factor} evaluates its argument into a new state $m'$ and a real value $v$, and returns a Dirac delta measure on the pair ($m'$, ()) weighted by $v$. 
To simplify the semantics, the type system ensures that the arguments of the probabilistic operators are always deterministic expressions.
To illustrate local declarations, \Cref{fig:sem:coit:kernel} shows the semantics of a simple expression with two local variables \zl{x} and \zl{y}. 
The state captures the previous value of the initialized variable $x$, and the state of all sub-expressions.
The transition function starts in a context where the previous value of $x$ is bound to a special variable $\attr{x}{last}$, and integrates over all possible executions of the sub-expressions to compute the main expression.

\paragraph{Inference}
So far, probabilistic expressions describe a stream of un-normalized measures over pairs (state, value).
At each step, the \zl{infer} operator normalizes the measure to obtain a distribution ($\top$ denotes the entire space), that is then split into a distribution of next states, and a distribution of values.
The corresponding stream of distributions is obtained by iteratively integrating the transition function along the distribution of states.

\begin{small}
\[
\mask{(\psemi{e}_{\gamma_0} =}\  \sigma_0 \mask{)}\ 
\mask{\xrightarrow{\int \sigma_0(dm) \psems{e}_{\gamma_1}(m)}} 
    \begin{array}[t]{c} \sigma_1\\\mu_1 \end{array}
\mask{\xrightarrow{\int \sigma_1(dm) \psems{e}_{\gamma_2}(m)}}
    \begin{array}[t]{c} \sigma_2\\\mu_2 \end{array}
\mask{\xrightarrow{\int \sigma_2(dm) \psems{e}_{\gamma_3}(m)}}
    \begin{array}[t]{c} \sigma_3\\\mu_3 \end{array}
\mask{\to} \ ...
\]
\end{small}

If the model is ill-defined, the normalization constant can be $0$ or $\infty$, which triggers an exception and stops the execution.
It is the programmer's responsibility to avoid such error cases when defining the model.

\subsection{Equations and fixpoints}
\label{sec:back:fixpoint}
In the interpretation of local declarations in \Cref{fig:sem:coit:kernel}, the nested integrals are only well defined if equations are ordered according to data dependencies.
The original semantics in~\cite{rppl-short} thus focuses on a  kernel language where local declarations are all \emph{scheduled}: initializations are grouped at the beginning and an equation $y = e_y$ must appear after $x = e_x$ if $x$ appears in $e_y$ outside a \zl{last}.
In the compiler, a specialized type system, the \emph{causality analysis} statically checks that a program is \emph{causal}, i.e., that all local declarations can be scheduled~\cite{CuoqP01}.
The kernel-based semantics is commutative, i.e., yields the same results for any valid schedule~\cite{staton17}, but imposing a scheduled order on equations is a significant limitation compared to block diagrams or synchronous dataflow languages which manipulate set of equations in arbitrary order.

The original co-iterative semantics~\cite{pouzet-cmcs98} and more recent works~\cite{emsoft23b} interpret mutually recursive equations using a fixpoint operator.
Values $v \in V$ are interpreted in a flat domain $V^{\bot} = V + \{\bot\}$ with minimal element $\bot$ and the flat order $\leq$: $\forall v \in V.\ \bot \leq v$.
$(V^{\bot}, \bot, \leq)$ is a complete partial order~(CPO).
This flat CPO is lifted to environments defining the same set of variables: $\forall \rho_1, \rho_2$ such that $\dom{\rho_1} = \dom{\rho_2} = X$, $\rho_1 \leq \rho_2$ iff $\forall x \in X,\  \rho_1(x) \leq \rho_2(x)$ and the least element is $\bot = [x \is \bot]_{x \in X}$.

\Cref{fig:sem:coit:det:eq} shows the semantics rules for deterministic equations adapted from~\cite{emsoft23b}.
The initial state of an equation is the initial state of its defining expression.
Given a state, the transition function returns a new state and an environment.
To interpret an expression with a set of local declarations $\zlwhere{e}{E}$, the transition function first computes the environment defined by $E$ with a fixpoint operator.
Given a state $M$, the function $F(\rho) = \letin{ M', \rho' = \sems{e}_{\gamma + \rho}(M)} \rho'$ is continuous and has a minimal fixpoint $\rho = \fix{F} = \lim_{n \to \infty}(F^n(\bot))$.
After convergence, the transition function evaluates $\sems{E}_{\gamma + \rho}(M)$ once more to compute the next state $M'$ (leaving $\rho$ unchanged by definition of the fixpoint) and finally evaluates the main expression $e$ in the environment $\gamma + \rho$.

If the program is \emph{causal} a valid schedule exists for $E$, and
by monotony, each fixpoint iteration computes the value of at least one variable and the fixpoint is reached after a finite number of iterations~\cite{emsoft23b}.

\begin{figure}
\begin{small}
\[
\begin{array}{lcl}
\sems{\zleq{x}{e}}_\gamma(m) &=& \letin{m', v = \sems{e}_\gamma(m)} m', [x \is v]
\\\\
\sems{\zlinit{x}{c}}_\gamma(p_x) &=& \gamma(x), [\attr{x}{last} \is p_x]
\\\\
\sems{\zland{E_1}{E_2}}_\gamma(M_1, M_2) &=&
\begin{stack}
    \letin{M_1', \gamma_1 = \sems{E_1}_\gamma(M_1)}\\
    \letin{M_2', \gamma_2 = \sems{E_2}_\gamma(M_2)}\\
    (M_1', M_2'), \gamma_1 + \gamma_2
\end{stack}
\\\\
\sems{\zlwhere{e}{E}}_\gamma(m, M) &=&
\begin{stack}
	\letin{F(\rho) = \left(\letin{M', \rho' = \sems{E}_{\gamma + \rho}(M)} \rho'\right)} \\
    \letin{\rho = \fix{F}}\\
	\letin{M', \rho = \sems{E}_{\gamma + \rho}(M)}\\
    \letin{m', v = \sems{e}_{\gamma + \rho}(m)}\\
    (m', M'), v
\end{stack}
\end{array}
\]
\end{small}
\caption{Co-iterative semantics of deterministic equations  with a fixpoint operator (based on~\cite{PouzetSynchron21}).}
\label{fig:sem:coit:det:eq}
\end{figure}

\section{Density-based co-iterative semantics}
\label{sec:coit}

In this section we detail the new density-based co-iterative semantics for probabilistic expressions.
We show that, in this semantics we can now interpret sets of mutually recursive equations with a fixpoint operator as in the original co-iterative semantics.
We then prove that the density-based semantics is equivalent to the kernel-based semantics, i.e., describes the same stream of distributions.

\subsection{Probabilistic co-iterative semantics with fixpoint}
\label{sec:sem:coit:dens}
The key idea of the density-based semantics is to externalize all sources of randomness.
Compared to the deterministic case, the transition function of a probabilistic expression takes one additional argument: an array of random seeds containing one random element for each random variable introduced by \zl{sample}.
To capture the effect of the \zl{factor} operator, the transition function also returns a weight which measures the quality of the result w.r.t. the model.

\paragraph{Expressions}
More formally, the initialization function of a probabilistic expression $e$, $\dsemi{e}_\gamma : S \times \mathbb{N}$ returns the initial state and the number of random variables in $e$.
Loops and recursive calls are not allowed in the language of \Cref{fig:syntax}.
The number of calls to \zl{sample} can thus be statically computed.
Given the current state and a value for all random seeds (an array of $p$ values in $[0, 1]$ where $p$ is the number of random variables computed by initialization function) the transition function $\dsems{e}_\gamma : \statetype \times [0, 1]^{p} \to \statetype \times \valuetype \times [0, \infty)$ returns a triple (next state, result, weight).

An excerpt of the density-based co-iterative semantics is presented in \Cref{sem:coit:expr}.
If $e$ is deterministic, there is no random variable and no conditioning.
The transition function takes an empty array of random seeds, evaluates the expression, and returns the next state, the value, and a weight of~$1$.
\zl{sample} defines one random variable.
The transition function takes an array containing one random seed, evaluates the argument into a distribution, converts the random seed into a sample of the distribution, and returns the next state, the sample, and a weight of~$1$.
\zl{factor} updates the weight.
The transition function evaluates it arguments into a real value $v$, and returns the next state, an empty value $()$, and the score $v$.
The initialization of a function call $\zlapp{f}{e}$ evaluates the initialization functions of $f$ and $e$, combines the initial states and sums the numbers of random variables.
The transition function takes an array containing the random seeds for $e$ and $f$,\footnote{We note $[r_1 : r_2]$ the concatenation of two arrays.} evaluates the argument $e$ into a value $v_e$ and a weight $w_e$, uses the value to evaluate the transition function of $f$ which returns a result $v$ and a weight $w_f$, and returns the combined next states, the result, and the total weight $w_e * w_f$.

\begin{figure}
\begin{small}
\[
\begin{array}{lll}
\dsemi{e}_\gamma &=& \semi{e}_\gamma, 0\\
\dsems{e}_\gamma(m, []) &=& 
    \letin{m', v = \sems{e}_\gamma(m)} 
    m', v, 1
    \qquad \text{if $e$ is deterministic}
\\\\
\dsemi{\zlsample{}{e}}_{\gamma} &=& \letin{m = \semi{e}_\gamma} m, 1\\
\dsems{\zlsample{}{e}}_{\gamma}(m, [r]) &=&
    \letin{m', \mu  = \sems{e}_\gamma(m)} 
    m', \sample{\mu}{r}, 1
\\\\
\dsemi{\zlfactor{e}}_{\gamma} &=&  \letin{m = \semi{e}_\gamma} m, 0\\
\dsems{\zlfactor{e}}_{\gamma}(m, []) &=& 
    \letin{m', v = \sems{e}_\gamma(m)} 
    m', (), v
\\\\
\dsemi{f(e)}_{\gamma} &=& 
    \begin{stack}
        \letin{m_f, p_f = \gamma(\attr{f}{init})}\\
        \letin{m_e, p_e = \dsemi{e}_\gamma}\\
        (m_f, m_e), p_f + p_e 
    \end{stack}\\
\dsems{f(e)}_{\gamma}((m_f, m_e), [r_f:r_e]) &=&
    \begin{stack}
        \letin{m_e', v_e, w_e = \dsems{e}_\gamma(m_e, r_e)}\\
        \letin{m_f', v, w_f = \gamma(\attr{f}{step})(v_e, m_f, r_f)}\\
        (m_f', m_e'), v, w_e * w_f
    \end{stack}
\end{array}
\]
\end{small}
\caption{Density-based co-iterative semantics for ProbZelus expressions (full version in \Cref{sem:coit:expr:full} of the appendix).}
\label{sem:coit:expr}
\end{figure}

\paragraph{Mutually recursive equations}
The semantics of probabilistic equations is presented in \Cref{sem:coit:eq}.
As for probabilistic expressions, the initialization function returns the initial states, and the number of random variables.
Given a state and an array of random seeds, the transition function returns a tuple (next state, environment, weight).
The equation $\zleq{x}{e}$ defines a single variable.
The transition function evaluates the defining expression $e$ into a tuple (next state, value, weight), and returns the next state, an environment where $x$ is bound to the value $v$, and the weight.
The $\zlinit{x}{e}$ equation manages the special variable $\zllast{x}$ which refers to the value of $x$ at the previous time step.
Compared to the original ProbZelus semantics, we do not require initial values to be constants. 
The state contains the previous value of $x$ initialized with an undefined value of the correct type $\nil$, and the initial state $m_0$ of the expression $e$.
There are two cases for the transition function.
At the first time step, or after a \zl{reset}, the state contains $\nil$ and
the transition function evaluates $e$ using $m_0$ to computes the initial value $i$ and the corresponding weight, and returns a new state containing the current value of $x$, an environment where $\attr{x}{last}$ is bound to $i$, and the weight.
In any other case, the previous value $v$ of $x$ stored in the state is defined. The transition function returns a new state containing the current value of $x$, an environment where $\attr{x}{last}$ is bound to $v$, and a weight of $1$.

Compared to the original kernel-based semantics described in \Cref{sec:back:co-it} which combines measures via integration, the density-based semantics only manipulates deterministic values for which the flat CPO on environments described in \Cref{sec:back:fixpoint} is well defined.
The initialization function of an expression with a set of local declarations $\zlwhere{e}{E}$ combines the initial states of $e$ and $E$ and returns the total number of random variables.
The transition function takes an array containing the random seeds for $e$ and $E$, computes the environment $\rho$ defined by $E$ with a fixpoint operator, evaluates $\dsem{E}_{\gamma+\rho}(M, r)$ once more to compute the next state $M'$ and the weight $W$, evaluates the main expression $e$ in the environment $\gamma + \rho$ which returns the result $v$ and a weight $w$, and returns the combined next states, the result, and the total weight $w * W$.
The only difference with the deterministic case is that the transition functions of $e$ and $E$ now take the random seeds as arguments and return the weights. 

\begin{figure}
\begin{small}
\[
\begin{array}{lll}
\dsemi{\zleq{x}{e}}_{\gamma} &=& \dsemi{e}_\gamma\\
\dsems{\zleq{x}{e}}_{\gamma}(m, r) &=&
    \letin{m',v, w = \dsems{e}_\gamma(m, r)}
     m',  [x \is v], w
\\\\
\dsemi{\zlinit{x}{e}}_{\gamma} &=& 
    \letin{m_0, p = \dsemi{e}_\gamma} 
    (\nil, m_0), p\\
\dsems{\zlinit{x}{e}}_{\gamma}((\nil, m_0), r) &=&
    \letin{m', i, w = \dsems{e}_\gamma(m_0, r)} 
    (\gamma(x), m_0), [\attr{x}{last} \is i], w\\
\dsems{\zlinit{x}{e}}_{\gamma}((v, m_0), r) &=&
    (\gamma(x), m_0), [\attr{x}{last} \is v], 1
\\\\
\dsemi{\zland{E_1}{E_2}}_{\gamma} &=& 
    \begin{stack}
        \letin{M_1, p_1 = \dsemi{E_1}_\gamma}\\
        \letin{M_2, p_2 = \dsemi{E_2}_\gamma}\\
        (M_1, M_2), p_1 + p_2
    \end{stack}\\
\dsems{\zland{E_1}{E_2}}_{\gamma}((M_1, M_2), [r_1: r_2]) &=&
    \begin{stack}
        \letin{M_1', \rho_1, w_1 = \dsems{E_1}_\gamma(M_1, r_1)}\\
        \letin{M_2', \rho_2, w_2 = \dsems{E_2}_{\gamma}(M_2, r_2)}\\
        (M_1', M_2'), \rho_1 + \rho_2, w_1 * w_2
    \end{stack}
\\\\
\dsemi{\zlwhere{e}{E}}_{\gamma} &=& 
    \begin{stack}
        \letin{m, p_e = \dsemi{e}_\gamma}\\
        \letin{M, p_E = \dsemi{E}_\gamma}\\ 
        (m, M), p_e + p_E
    \end{stack}\\
\dsems{\zlwhere{e}{E}}_{\gamma}((m, M), [r_e : r_E]) &=& 
    \begin{stack}
        \letin{F(\rho) = \left(
            \letin{M', \rho, w = \dsem{E}_{\gamma + \rho}(M, r_E)} \rho \right)}\\   
        \letin{\rho = \fix{F}}\\
        \letin{M', \rho, W = \dsems{E}_{\gamma + \rho}(M, r_E)}\\
        \letin{m', v, w = \dsems{e}_{\gamma + \rho}(m, r_e)}\\
        (m', M'), v, w * W
    \end{stack}
\end{array}
\]
\end{small}
\caption{Density-based co-iterative semantics for ProbZelus equations (full version in \Cref{sem:coit:eq:full} of the appendix).}
\label{sem:coit:eq}
\end{figure}
\paragraph{Scheduling}
To compare the density-based semantics with the kernel-based semantics, it is useful to define an alternative semantics for a scheduled language without a fixpoint operator.
This alternative semantics $\dssems{e}{}$ exactly matches the density-based semantics except for the two following rules:

\[
\begin{small}
\begin{array}{lll}
    \dssems{\zland{E_1}{E_2}}\gamma((M_1, M_2), [r_1:r_2]) &=&
        \begin{stack}
            \letin{M_1', \rho_1, w_1 = \dssems{E_1}\gamma(M_1, r_1)}\\
            \letin{M_2', \rho_2, w_2 = \dssems{E_2}{\gamma + \rho_1}(M_2, r_2)}\\
            (M_1', M_2'), \rho_1 + \rho_2, w_1 * w_2
        \end{stack}
    \\\\
    \dssems{\zlwhere{e}{E}}\gamma((m, M), 
    [r_e : r_E]) &=&
        \begin{stack}
            \letin{M', \rho, W = \dssems{E}\gamma(M, r_E)}\\
            \letin{m', v, w = \dssems{e}{\gamma + \rho}(m, r)}\\
            (m', M_1', M_2'), v, w * W
        \end{stack}
\end{array}
\end{small}
\]

\medskip
\noindent
Since all equations are scheduled, the environment produced by a set of equations can be computed incrementally and there is no need for a fixpoint operator to interpret local declarations.

\begin{proposition}
    \label{prop:coit:sched}
    For an expression where all equations are scheduled, the scheduled density-based semantics is equal to the density-based semantics with a fixpoint.
\end{proposition}

\begin{proof}
    This result is a consequence of the following lemma:

    \begin{lemma}
    For all scheduled equations set $E$, the scheduled semantics yields the same environment as the fixpoint operator, i.e., for an environment $\gamma$, a state $M$ and an array of random seeds $r$:
    \[
        \fix{\lambda \rho. \ \letin{M', \rho', w = \dsems{E}_{\gamma+\rho}(M, r)} \rho'}
        \ \ = \ \  \letin{M', \rho, w = \dssems{E}{\gamma}(M, r)} \rho
    \]
    \end{lemma}

    \noindent
    The proof is by induction on $E$.
    It is sufficient to focus on the case $\zland{E_1}{E_2}$.
    Since equations are scheduled, $E_1$ does not depend on variables defined in $E_2$ and we have for an environment $\gamma$, a state $(M_1, M_2)$ and an array of random seeds $[r_1:r_2]$:

    \[
    \begin{small}
    \begin{array}{ll}
        &\mit{fix}(\lambda (\rho_1 + \rho_2). \ 
        \begin{stack}
            \letin{M_1', \rho_1', w_1 = \dsems{E_1}_{\gamma + \rho_1 + \rho_2}(M_1, r_1)}\\
            \letin{M_2', \rho_2', w_2 = \dsems{E_2}_{\gamma + \rho_1 + \rho_2}(M_2, r_2)}
            \rho_1' + \rho_2')
        \end{stack}\\
        = &
        \mit{fix}(\lambda (\rho_1 + \rho_2). \ 
        \begin{stack}
            \letin{M_1', \rho_1', w_1 = \dsems{E_1}_{\gamma + \rho_1}(M_1, r_1)}\\
            \letin{M_2', \rho_2', w_2 = \dsems{E_2}_{\gamma + \rho_1 + \rho_2}(M_2, r_2)}
            \rho_1' + \rho_2')
        \end{stack}\\
        = &
        \begin{stack}
            \letin{\rho_1'' = \fix{\lambda \rho_1. \ \letin{M_1', \rho_1', w_1 = \dsems{E_1}_{\gamma + \rho_1}(M_1, r_1)} \rho_1'}}\\
            \letin{\rho_2'' = \fix{\lambda \rho_2. \ \letin{M_2', \rho_2', w_2 = \dsems{E_2}_{\gamma + \rho_1'' + \rho_2}(M_2, r_2)} \rho_2'}}
            \rho_1'' + \rho_2''
        \end{stack}\\
        = &
        \begin{stack}
            \letin{M_1', \rho_1', w_1 = \dssems{E_1}{\gamma}(M_1, r_1)}\\
            \letin{M_2', \rho_2', w_2 = \dssems{E_2}{\gamma + \rho_1'}(M_2, r_2)}
            \rho_1' + \rho_2'
        \end{stack}
    \end{array}
    \end{small}
    \]
\end{proof}

\subsection{Inference}
\label{sec:coit:inf}

The \zl{infer} operator first turns the result of the density-based semantics into an un-normalized measure, and then performs the same operation as in the kernel-based semantics: 1) integrate over all possible states, 2) normalize the measure, 3) split the result into a distribution of next states and a distribution of values.

\begin{equation}\label{eq:infer:coit}
\begin{small}
\begin{array}{lll}
\semdi{\zlinfer{e}}_{\gamma} &=& \letin{m, p = \dsemi{e}_\gamma} \delta_m, p\\[1em]
\semds{\zlinfer{e}}_{\gamma}(\sigma, p) &=&
    \begin{stack}
    \letin{\psi(m) = \displaystyle \int_{[0, 1]^p} 
        \letin{m', v, w = \dsems{e}_\gamma(m, r)} w * \delta_{(m', v)}\ dr}\\
    \letin{\nu = \displaystyle \int \sigma(dm) \ \psi(m)}\\
    \letin{\overline{\nu} = \nu / \nu(\top)}\\
    (\pi_{1*}(\overline{\nu}), p), \pi_{2*}(\overline{\nu})
    \end{stack}
\end{array}
\end{small}
\end{equation}

\noindent
The state of the \zl{infer} operator contains the number of random variables in the model $p$ (which remains constant) and a distribution of possible states.
The initial distribution of states is a Dirac delta measure over the initial state of the model.
The transition function first computes a function $\psi$ mapping a state to the un-normalized measure which associates each pair (next state, value) to its weight.
The \zl{infer} operator then integrates this function along all possible values of the state, normalizes it, and splits it into a pair of distributions.

\paragraph{Correctness}
The previous definition is very similar to its kernel-based semantics counterpart where the function $\psi(m)$ in~\Cref{eq:infer:coit} plays  the role of the semantics of the model.
We now show that these two notions coincide.

\begin{proposition}
\label{thm:density-kernel}
For all probabilistic expression $e$ with $p$ random variables where all equations are scheduled, the density-based semantics is the density of the measure computed by the kernel semantics, i.e., for all environment $\gamma$ and state $m$:
\[
\int_{[0, 1]^{p}} \letin{m', v, w = \dsems{e}_\gamma(m)} w*\delta_{m', v} \ dr
=
\psems{e}_\gamma(m)   
\]
\end{proposition}

\begin{proof}
The kernel-based semantics is only defined for a scheduled language.
We first prove by induction on the structure of $e$ that the scheduled density-based semantics coincide with the kernel-based semantics.
We can then conclude with \Cref{prop:coit:sched}.

The case $\zlsample{}{\mu}$ is a simple variable substitution $x=\sample{\mu}{r}$ where $\mathit{icdf}_{\mu}$ is the inverse of the cumulative function of $\mu$. Indeed, any real continuous distribution $\mu$  is the pushforward by $\mathit{icdf}_{\mu}$ of  the uniform distribution over $[0, 1]$ denoted $\lambda$:
\begin{small}
\begin{align*}
\int_{[0, 1]} \delta_{\sample{\mu}{r}}\ dr 
= \int \delta_{\sample{\mu}{r}}\ \lambda(dr) 
= \int \delta_{x} \ \mathit{icdf}_{\mu*}(\lambda)(dx) 
= \int \delta_{x}\ \mu(dx)
= \mu
\end{align*}
\end{small}
This property generalizes to discrete distributions, multivariate distributions, and any distributions over Polish spaces. 
By analogy, we use the notation $\mathit{icdf}_\mu$ in all cases.

The case $\zland{E_1}{E_2}$ is a consequence of Fubini's theorem.

\begin{small}
\[
\begin{array}[t]{l}
\int \psems{E_1}_\gamma(M_1)(dM_1', d \rho_1) \int \psems{E_2}_{\gamma + \rho_1}(M_2)(dM_2', d \rho_2)\ \delta_{(M_1', M_2'), \rho_1 + \rho_2}\\
\quad \begin{array}[t]{l} 
= \int \left( \int_{[0, 1]^{p_1}} \letin{M_1', \rho_1, w_1 = \dsems{E_1}_\gamma(M_1)} w_1 * \delta_{M_1, \rho_1}\ d r_1 \right)(dM_1', d \rho_1)\\
\quad \int \left( \int_{[0, 1]^{p_2}} \letin{M_1', \rho_2, w_1 = \dsems{E_2}_{\gamma + \rho_1}(M_2)} w_2 * \delta_{M_2, \rho_2}\ d r_2 \right)(dM_2', d \rho_2)\\
\qquad \delta_{(M_1', M_2'), \rho_1 + \rho_2}\\
= \int_{[0, 1]^{p_1}} \int_{[0, 1]^{p_2}} \int \left(  \letin{M_1', \rho_1, w_1 = \dsems{E_1}_\gamma(M_1)} w_1 * \delta_{M_1, \rho_1} \right)(dM_1', d \rho_1)\\
\qquad \qquad \qquad \ \  \int \left(  \letin{M_1', \rho_2, w_1 = \dsems{E_2}_{\gamma + \rho_1}(M_2)} w_2 * \delta_{M_2, \rho_2} \right)(dM_2', d \rho_2)\ d r_1 d r_2\\
\qquad \qquad \qquad \qquad \delta_{(M_1', M_2'), \rho_1 + \rho_2}\\[0.5em]
= \int_{[0, 1]^{p_1 + p_2}} \begin{array}[t]{l}
     \letin{M_1', \rho_1, w_1 = \dsems{E_1}_\gamma(M_1)}\\
    \letin{M_2', \rho_2, w_2 = \dsems{E_2}_{\gamma + \rho_1}(M_2)}\\ 
    w_1 * w_2 * \delta_{(M_1', M_2'), \rho_1 + \rho_2}\ d r_1 d r_2
\end{array}
\end{array}
\end{array}
\]
\end{small}

Other cases are similar.
\end{proof}

We can now state the main correctness theorem, i.e., the \zl{infer} operator yields the same stream of distributions in the density-based semantics and in the kernel based semantics.

\begin{theorem}[Co-iterative semantics correctness]
For all probabilistic model $e$ where all equations sets are scheduled, for all environment $\gamma$, and for all distribution of state $\sigma$:

\[
\begin{array}{lll}
    \semdi{\zlinfer{e}}_\gamma &=& \sems{\zlinfer{e}}_\gamma, p\\
    \semds{\zlinfer{e}}_\gamma(\sigma, p) &=&\sems{\zlinfer{e}}_\gamma(\sigma) 
\end{array}
\]
\end{theorem}

\begin{proof}
By construction, in the density-based semantics, the first element of the initial state of \zl{infer} is a Dirac delta measure on the initial state of the model which corresponds to the initial state of \zl{infer} in the kernel-based semantics.

By \Cref{thm:density-kernel} the un-normalized measure defined by the density-based semantics matches the measure computed by the kernel-based semantics.
Given this measure, the rest of the transition function of \zl{infer} is the same in both cases.
\end{proof}

\section{Density-Based Relational Semantics}
\label{sec:rel}

An alternative to the operational view of the co-iterative semantics where expressions are interpreted as state machines is to define a relational semantics where expressions directly return streams of values~\cite{ColacoP03}.
This formalism has been used in the Vélus project to prove an end-to-end dataflow synchronous compiler within the Coq proof assistant~~\cite{BourkeBDLPR17, BourkeBP20, BourkeJPP21}.

In this section, we first present a relational semantics for the deterministic expressions of our language.
We then define a relational density-based semantics for probabilistic expressions and prove that this semantics is equivalent to the co-iterative density-based semantics, i.e., the \zl{infer} operator yields the same stream of distributions.

\paragraph{Notations}
In the following, $\Stream{V}$ is the type of infinite streams of values of type $V$.
The infix operator $(\cdot) : V \to \Stream{V} \to \Stream{V}$ is the stream constructor (e.g., $1 \cdot 2 \cdot 3 \cdot ...$). 
Constants are lifted to constant streams (e.g., $1 = 1 \cdot 1 \cdot 1 \cdot ...$) and when the context is clear we write $f(s) = f(s_0) \cdot f(s_1) \cdot ...$ for $\map{f}{s}$, and $(s, t) = (s_1, t_1) \cdot (s_2, t_2) \cdot ...$ to cast a pair of streams into a stream of pairs.

\subsection{Deterministic relational semantics}
\label{sec:sem:rel:dens}

In the relational semantics, deterministic expressions compute streams of values.
In a context $H$ which maps variables names to stream of values, the semantics of a deterministic expression $e$ returns a stream $s$: $\rsemexp{G, H}{e}{s}$.
The additional context $G$ stores global declarations (global constants and function definitions).
The semantics of a set of equations $E$ checks that the context $H$ is compatible with all the equations: $\rsemeq{G, H}{E}$.
The semantics of a set of equations thus defines a \emph{relation} between the streams stored in the context.
Compared to the co-iterative semantics, the relational semantics is not executable since the context must be guessed a priori and validated against the equations.

\begin{figure}
\begin{small}
\begin{mathpar}
\inferrule
    {}
    {G, H \vdash c \downarrow c}
\and
\inferrule
    {}
    {G, H \vdash x \downarrow H(x)}
\and
\inferrule
    {G, H \vdash e_1 \downarrow s_1\\ 
     G, H \vdash e_2 \downarrow s_2}
    {G, H \vdash \zlpair{e_1}{e_2} \downarrow (s_1, s_2)}
\and
\inferrule
    {G, H \vdash e \downarrow s}
    {G, H \vdash \zlop{e} \downarrow \op(s)}
\and
\inferrule
    {H(\attr{x}{last}) = s \\ }
    {G, H \vdash \zllast{x} \downarrow s}
\and
\inferrule
    {G, H \vdash e \downarrow s_e\\
     G(f) = \zlnode{f}{x}{e_f} \\
     G, [x \is s_e] \vdash e_f \downarrow s}
    {G, H \vdash f(e) \downarrow s}
\and
\inferrule
    {G, H + H_E \vdash E \\ 
        G, H + H_E \vdash e \downarrow s}
    {G, H \vdash \zlwhere{e}{E} \downarrow s}

\and
\inferrule
    {G, H \vdash e \downarrow H(x)}
    {G, H \vdash \zleq{x}{e}}
\and
\inferrule
    {G, H  \vdash e \downarrow i \cdot s \\ 
        H(\attr{x}{last}) = i \cdot H(x)}
    {G, H  \vdash \zlinit{x}{e}}
\and
\inferrule
    {G, H \vdash E_1 \\
     G, H \vdash E_2}
    {G, H \vdash \zland{E_1}{E_2}}
\end{mathpar}
\end{small}
\caption{Deterministic relational semantics (full version in \Cref{fig:sem:rel:det:full} of the appendix).}
\label{fig:sem:rel:det}
\end{figure}

\Cref{fig:sem:rel:det} presents the relational semantics for deterministic expressions and equations.
A constant is interpreted as a constant stream, and a variable returns the corresponding stream in the context.
The semantics of a pair evaluates each component independently and packs the results into a stream of pairs.
The application of an operator evaluates its argument into a stream of values and maps the operator on the result.
$\zllast{x}$ fetches a special variable $\attr{x}{last}$ in the context.
A function call first evaluates its argument, and then evaluates the body of the function in a context where the parameter is bound to the argument value.

To interpret an expression with a set of local declarations $\zlwhere{e}{E}$, equations $E$ are evaluated in a new context $H_E$ that is also used to evaluate the main expression $e$.
The semantics of a simple equation checks that a variable is associated to the stream computed by its defining expression.
The initialization operator $\zlinit{x}{e}$ prepends an initial value $i$ to the stream associated to $x$ and checks that the special variable $\attr{x}{last}$ is bound to this delayed version of $x$.
In the relational semantics, context are un-ordered maps and scheduling equations does not change the semantics.

\subsection{Probabilistic relational semantics}

The key idea of the probabilistic relational semantics is similar to the density-based co-iterative semantics: instead of manipulating streams of measures, probabilistic expressions compute streams of pairs (value, score) using external streams of random seeds, and integration is deferred to the \zl{infer} operator.

\Cref{fig:sem:rel:prob} presents the density-based relational semantics for probabilistic expressions and equations.
In a context $H$ which maps variables names to values, the semantics of a probabilistic expression $e$ takes an array of random streams $R$ and returns a stream of pairs (value, weight): $\prsemexp{G, H, R}{e}{(s, w)}$.
The semantics of a set of equations $E$ takes an array containing the random streams of all sub-expressions, checks that the contexts $H$ is compatible with all the equations, and returns the total weight $W$ of all sub-expressions: $\prsemeq{G, H, R}{E}{W}$.

The semantics of deterministic expressions (e.g., constant or variable) returns the expected stream of values associated to a constant weight of $1$.
The semantics of \zl{sample} takes an array containing one random stream $R$, evaluates its argument into a stream of distributions $s_\mu$, and uses the random stream to compute a stream of samples associated to the constant weight~$1$: $(\sample{{s_\mu}_0}{R_0}, 1) \cdot (\sample{{s_\mu}_1}{R_1}, 1) \cdot ...$.
The semantics of \zl{factor} evaluates its arguments into a stream of values $w$ which is used as the weight associated to a stream of empty values: $((), w_0) \cdot ((), w_1) \cdot ...$.
The semantics of a function call is similar to the deterministic case, but the random streams are split between the argument and the function body, and the total weight captures the weight of the argument and the weight of the function body.
Similarly, for an expression with a set of local definitions the random streams are split between sub-expressions and the weight is the total weight of all sub-expressions.

By construction, for any probabilistic expression $e$, the size of the array of random streams is the number of random variables defined in $e$, i.e., the number of \zl{sample}.
This information can be computed statically (as in the initialization functions of the co-iterative semantics in \Cref{sec:sem:coit:dens}), and in the following $\RV(e)$ returns the number of random variables in $e$. 

\begin{figure}
\begin{small}
\begin{mathpar}
\inferrule
    {G, H \vdash e \downarrow s}
    {G, H, [] \vdash e \Downarrow (s, 1)}
\and

\inferrule
    {G, H \vdash e \downarrow s_\mu}
    {G, H, [R] \vdash \zlsample{\alpha}{e} \Downarrow (\sample{s_\mu}{R}, 1)}
\and
\inferrule
    {G, H \vdash e \downarrow w}
    {G, H, [] \vdash \zlfactor{e} \Downarrow ((), w)}
\and
\inferrule
    {G, H, R_e \vdash e \downarrow (s_e, w_e)\\
     G(f) = \zlproba{f}{x}{e_f} \\ 
     G, [x \is s_e], R_f \vdash e_f \Downarrow (s, w)}
    {G, H, [R_e:R_f] \vdash \zlapp{f}{e} \Downarrow (s, w * w_e)}
\and
\inferrule
    {G, H + H_E, R_E \vdash E : w_E \\ 
        G, H + H_E, R_e \vdash e \Downarrow (s, w)}
    {G, H, R_E \vdash \zlwhere{e}{E} \Downarrow (s, w * w_E)}

\and
\inferrule
    {G, H, R \vdash e \Downarrow (H(x), w)}
    {G, H, R \vdash \zleq{x}{e} : w}
\and
\inferrule
    {G, H, R \vdash e \Downarrow (i \cdot s, w_i \cdot w) \\ 
        H(\attr{x}{last}) = i \cdot H(x)}
    {G, H, R \vdash \zlinit{x}{e} : w_i \cdot 1}
\and
\inferrule
    {G, H, R_1 \vdash E_1 : w_1 \\
     G, H, R_2 \vdash E_2 : w_2}
    {G, H, [R_1:R_2] \vdash \zland{E_1}{E_2} : w_1 * w_2}
\and
\inferrule
    { p = \RV(e)\\
      \left[G, H, R \vdash e \Downarrow (s, w) \qquad \overline{w} = \Pi\ w \right]_{R \in ([0,1]^\omega)^p}}
    {G, H \vdash \zlinfer{e} \downarrow \integ{p}{\overline{w}}{s}}
\end{mathpar}
\end{small}
\caption{Probabilistic relational semantics (full version in \Cref{fig:sem:rel:prob:full} of the appendix).}
\label{fig:sem:rel:prob}
\end{figure}

\subsection{Inference}
\label{sec:sem:rel:inf}

As in the density-based co-iterative semantics, the \zl{infer} operator is defined by integrating at each step an un-normalized density function over all possible values of the streams of random seeds.
The semantics of a probabilistic model returns a pair of stream functions (value, weight) which both depend on the random streams.
Given the random streams, at each time step, the semantics of \zl{infer} first computes the total weight of the prefix to capture all the conditioning since the beginning of the execution: $ \overline{w}(R) = \Pi \ w(R) = w_0(R) \cdot (w_0(R) * w_1(R)) \cdot (w_0(R) * w_1(R) * w_2(R)) \cdot ...$.
Then the function $\mit{integ}$ 1) turns $v_n$ and $\overline{w_n}$ into an un-normalized measures by integrating over all possible values of the random streams, and 2) normalizes the result to obtain a stream of distributions of values.
If $p = \RV(e)$ is the number of random variables in the model and $\lambda^p_\omega$ is the uniform measure over the cube of random streams $([0, 1]^\omega)^{p}$:
\begin{equation}\label{eq:rel:infer}
\begin{array}{lcl}

\integ{p}{(\overline{w} \cdot \overline{ws})}{(v \cdot vs)} &=& 
\left(
\letin{\mu = \int \overline{w}(R) \delta_{v(R)} \ \lambda^p_\omega(dR)} \mu/\mu(\top) \right) 
\cdot 
(\integ{}{ws}{vs})
\end{array}
\end{equation}

\paragraph{Cube of random streams}
The uniform measure over the cube of random streams is defined as follows.
Let $\Stream{[0,1]}$ be the countable product of the measurable spaces on the interval $[0, 1]$ endowed with the Lebesgue $\sigma$-algebra, i.e., the coarsest $\sigma$-algebra such that projections are measurable.
We define $\lambda_\omega$ as the uniform distribution on the \emph{continuous cube}
defined by a Kolmogorov extension such that for any $k \in \mathbb{N}$, the pushforward measure of $\lambda_\omega$ along the projection $\pi_{\leq k}: [0,1]^\omega\to [0,1]^k$ on the first $k$ coordinates is the Lebesgue measure on $[0, 1]^k$: $\lambda_{\leq k} = {\pi_{\leq k}}_{*}(\lambda_\omega)$.
For any measurable function $g: [0, 1]^k \to V$ we have the following change of variable formula:
\[
    \int g(\pi_{\leq k}(R))\ \lambda_\omega(dR) =\int g(R_{\leq k})\ \lambda_{\leq k}(dR_{\leq k})
\]
Integrating a function which only depends on the $k$ first coordinates of $R$ can thus be reduced to integrating over these coordinates.
We can then define the uniform measure on the cube of random streams $\lambda^p_\omega$ as the $p$-ary product measure of $\lambda_\omega$, and lift the change of variable formula.

\paragraph{Correctness}
For a probabilistic expression $e$, we first relate the relational semantics of \Cref{sec:sem:coit:dens} and the co-iterative semantics of \Cref{sec:sem:rel:dens}.
If $H$ is an environment mapping variables names to streams of values, $H_k$ is the environment where streams are projected on their $k$-th coordinate and $H_{\leq k}$ is the environment where streams are truncated at $k$.
We define similarly $R_{\leq k}$, and $R_k$ for an array of random streams~$R$.

\begin{proposition}
    For a causal probabilistic model $e$,
    if ${G, H, R \vdash e \Downarrow (s,w)}$ there is a co-iterative execution trace $m_0=\dsemi{e}_G$  and 
      ${\forall k>0,\ (m_{k+1}, v_{k+1}, w'_{k+1})=\dsems{e}_{H_{k+1}}(m_k, R_{k+1})}$ such that $\forall k >0,$
    $m_k$, $v_k$, $w'_k$ only depend on $H_{\leq k}$ and $R_{\leq k}$, and 
    $s_{k}(H)=v_{k}(H_{\leq k},R_{\leq k})$ and $w_{k}(H)=w'_{k}(H_{\leq k},R_{\leq k})$. 
\end{proposition}

\begin{sketch}
    This proposition states that if a program is causal, i.e., if all equations can be scheduled, the co-iterative semantics and the relational semantics coincide.

    First, we can show that if a relational semantics exists, there exists a state machine whose execution matches the relational semantics. 
    A similar proof is at the heart of the Vélus compiler~\cite{BourkeBDLPR17}.
    While adapting this proof to the ProbZelus language with explicit scores is far from easy, it does not offer any new insights.
    Second, the semantics of the compiled state-machines is deterministic and corresponds to the co-iterative semantics of the normalized scheduled program which does not require any fixpoint computation.
    Finally, the co-iterative semantics can be used to prove that source-to-source transformations preserve the semantics, in particular the normalization and scheduling passes~\cite{emsoft23b}.
    The semantics of the state machine thus corresponds to the co-iterative semantics of the original program.
        
    The property also states that at each instant, the output of a causal model only depends on past inputs and states which can be proved by induction on the structure of the program.
\end{sketch}

\medskip

As in \Cref{sec:coit:inf}, we can now state the main correctness theorem, i.e., the \zl{infer} operator yields the same stream of distributions in the co-iterative semantics and in the relational semantics.

\begin{theorem}[Relational semantics correctness]\label{thm:rel:bisim}
For a causal probabilistic model $e$, and for all environments $G, H$, if $\ {G, H \vdash \zlinfer{e} \downarrow \mu}$ then the co-iterative execution trace yields the same stream of distributions, i.e., $\sigma_0, p =\semi{\zlinfer{e}}_G$ and $\forall k > 0, (\sigma_k, p), \mu_k = \sems{\zlinfer{e}}_{H_{k+1}}(\sigma_k, p)$.
\end{theorem}

\begin{proof}    
    If $p$ is the number of random variables in the model, we show $\forall k > 0$: 
    
    \begin{small}
    \begin{align*}
        \sigma_{k+1}(H_{k+1}, \sigma_k) \propto \int_{([0, 1]^k)^p} \overline{w_k}(H_{\leq k}, R_{\leq k}) * \delta_{m_k(H_{\leq k}, R_{\leq k})} \lambda^p_{\leq k}(d R_{\leq k})\\
        \mu_{k+1}(H_{k+1}, \sigma_k) \propto \int_{([0, 1]^k)^p}  \overline{w_k}(H_{\leq k}, R_{\leq k}) * \delta_{v_k(H_{\leq k}, R_{\leq k})} \lambda^p_{\leq k}(d R_{\leq k})
    \end{align*}
    \end{small}

    \noindent
    By the induction hypothesis, $\forall k$, $H$, $R$, $s(H,R)_k=v_k(H_{\leq k}, R_{\leq k})$ and $w(H,R)_k= w'_k(H_{\leq k}, R_{\leq k})$.
    From the definition of $\psi$ in \Cref{eq:infer:coit} and Fubini's theorem we have:

    \begin{small}
    \begin{align*}
        \nu_{k+1} &= \displaystyle \int \sigma_k(dm) \ \psi(m)\\
        &\propto \int_{([0,1]^{k})^p}
        \overline{w_k}(H_{\leq k},R_{\leq k}) * \psi(m_k(H_{\leq k},R_{\leq k}))\ \lambda^n_{\leq k}(d R_{\leq k})\\
       &\propto \int_{([0,1]^{k})^p} \int_{[0, 1]^p}
       \overline{w_{k}}(H_{\leq k},R_{\leq k}) * w_{k+1}(H_{k+1},R_{k+1})\\[-0.75em]
       & \qquad \qquad \qquad \qquad \quad * \delta_{m_{k+1}(H_{\leq k+1},R_{\leq k+1}), v_{k+1}(H_{\leq k+1},R_{\leq k+1})}\ \lambda^p(d R_{k+1})\lambda^p_{\leq k}(d R_{\leq k})\\
       &\propto \int_{([0,1]^{k+1})^p}
       \overline{w_{k+1}}(H_{\leq k+1},R_{\leq k+1}) * \delta_{m_{k+1}(H_{\leq k+1},R_{\leq k+1}), v_{k+1}(H_{\leq k+1},R_{\leq k+1})}\ \lambda^p_{\leq k+1}(d R_{\leq k+1})
    \end{align*}
    \end{small}

    \noindent
    The normalization and marginalization by $\pi_{1*}$ and $\pi_{2*}$ concludes the inductive case.
    Then using the change of variable formula on the cube of random streams we get:
    \[
        \begin{small}
        \mu_k \propto \int \overline{w_k}(R)\delta_{s_k(R)}\lambda^p_\omega(dR)
        \end{small}
    \]
    which corresponds to \Cref{eq:rel:infer} and concludes the proof.
\end{proof}

\subsection{Program equivalence}
Compared to the co-iterative semantics where proving the equivalence between two state machines requires a bisimulation, in the relational semantics, to prove the equivalence between two programs one need only check that they define the same streams.

\begin{definition}
Deterministic expressions $e_1$ and $e_2$ are equivalent if for all contexts $G, H$: \[G, H \vdash e_1 \downarrow s \quad \text{and} \quad G, H \vdash e_2 \downarrow s\]
\end{definition}

Two probabilistic expressions are equivalent if they describe the same stream of measures obtained by integrating at each step the streams of pairs (value, weight) computed by the density-based relational semantics. 

\begin{definition}
    Probabilistic expressions $e_1$ and $e_2$ with $\RV(e_1) = p_1$ and $\RV(e_2) = p_2$ are equivalent if for any context $G, H$:
    \[
    \forall k > 0. \ 
    \int {\overline{w_1}}_k(R_1) * \delta_{{s_1}_k(R_1)}\ d \lambda^{p_1}_\omega(R_1) 
    = 
    \int {\overline{w_2}}_k(R_2) * \delta_{{s_2}_k(R_2)}\ d \lambda^{p_2}_\omega(R_2)
    \]
    where for all random streams $R_1, R_2$: $G, H, R_1 \vdash e_1 \Downarrow (s_1, w_1)$ and $G, H, R_2 \vdash e_2 \Downarrow (s_2, w_2)$.
\end{definition}

In the relational semantics, for a given context, each pair (value, weight) is a function of the random streams.
Since, random streams are uniformly distributed, if we can map the random streams of expression $e_1$ to the random streams of expression $e_2$ while preserving uniform distributions, program equivalence can be reduced to the comparison of the streams of pairs (value, weight) computed by each expression.

\begin{proposition}[Probabilistic equivalence]
    \label{prop:equiv}
    Probabilistic expressions $e_1$ and $e_2$ with $\RV(e_1) = p_1$ and $\RV(e_2) = p_2$ are equivalent if there is a measurable function $f: ([0, 1]^\omega)^{p_1} \to ([0, 1]^\omega)^{p_2}$ such that, for all contexts $G, H$:
    \begin{itemize}
    \item if $f(R_1) = R_2$ and $G, H, R_i \vdash e_i \Downarrow (s_i, w_i)$, then $s_1(R_1) = s_2(R_2)$ and $w_1(R_1) = w_2(R_2)$, and
    \item $\lambda^{p_2}_\omega$ is the pushforward of $\lambda^{p_1}_\omega$ along $f$, i.e., $\lambda^{p_2}_\omega = f_*(\lambda^{p_1}_\omega)$.
    \end{itemize}
\end{proposition}

\begin{proof}
The proof is a direct application of the change of variable formula.
For all contexts $G, H$, and $k >0$:
\begin{small}
\begin{align*}
    \int {\overline{w_2}}_k(R_2) * \delta_{{s_2}_k(R_2)}\ \lambda^{p_2}_\omega(dR_2)
    &= \int {\overline{w_2}}_k(f(R_1)) * \delta_{{s_2}_k(f(R_2))}\ f_*(\lambda^{p_1}_\omega)(dR_2)\\
    &=\int {\overline{w_1}}_k(R_1) * \delta_{{s_1}_k(R_1)}\ \lambda^{p_1}_\omega(dR_1) 
    \end{align*}
\end{small}
\end{proof}

Finding such a mapping is in general difficult.
A useful simple case is when the two programs involve the same random variables in different orders, e.g., a program and its compiled version after a source-to-source transformation.
In this case, the measurable function is a permutation of the random streams, and two expressions are equivalent if they compute the same stream of pairs (value, weight).

The relational semantics of an expression is described by a derivation tree where each relation is a consequence of smaller relations on all the sub-expressions, up to atomic expressions.
Two expressions compute the same streams if from the derivation tree of the first, one can build a derivation tree for the second and vice-versa.

\begin{example*} If $x$ and $y$ are not free variables in expressions $e_1$ and $e_2$:
\[ \zlsample{}{e_1}+\zlsample{}{e_2} \sim \zlwhere{x+y}{\zland{\zleq{x}{\zlsample{}{e_2}}}{\zleq{y}{\zlsample{}{e_1}}}}\]

\noindent
Let $R_i$ be the random streams associated to the expressions $\zlsample{}{e_i}$.
For all context $G, H$, if $G, H \vdash e_i \downarrow \mu_i$ we define $s_i = \sample{\mu_i}{R_i}$.
Then, the derivation tree for the lhs expression is:

\begin{footnotesize}
\[
\begin{prooftree}
\hypo{G, H, R_1 \vdash \zlsample{}{e_1} \Downarrow (s_1, 1)}
\hypo{G, H, R_2 \vdash \zlsample{}{e_2} \Downarrow (s_2, 1)}
\infer2{
    G, H, [R_1: R_2] 
    \vdash \zlsample{}{e_1}+\zlsample{}{e_2}
    \Downarrow (s_1 + s_2, 1)
}
\end{prooftree}
\]
\end{footnotesize}
\medskip

\noindent
With $H_E = [x \is s_2, y \is s_1]$, the derivation tree for the rhs expression is:

\begin{footnotesize}
\[
\begin{prooftree}
\hypo{
    G, H + H_E, [] \
    \vdash x + y \Downarrow (s_2 + s_1,1) 
}
\hypo{
    G, H + H_E, R_2
    \vdash \zlsample{}{e_2}
    \Downarrow (s_2, 1)
}
\infer1{
    G, H + H_E, R_2
    \vdash \zleq{x}{\zlsample{}{e_2}}
    : w_2 
}
\hypo{
    G, H + H_E, R_1
    \vdash \zlsample{}{e_1}
    \Downarrow (s_1, 1)
}
\infer1{
    G, H+ H_E, R_1 
    \vdash \zleq{y}{\zlsample{}{e_1}}
    : w_1 
}
\infer2{
    G, H + H_E, [R_2:R_1]
    \vdash 
        \zland
            {\zleq{x}{\zlsample{}{e_2}}}
            {\zleq{y}{\zlsample{}{e_1}}}
    : w_1 * w_2}
\infer2{
    G, H, [R_2: R_1] 
    \vdash 
        \zlwhere
            {x+y}
            {\zland
                {\zleq{x}{\zlsample{}{e_2}}}
                {\zleq{y}{\zlsample{}{e_1}}}}
    \Downarrow (s_2 + s_1, 1)
}
\end{prooftree}
\]
\end{footnotesize}
\medskip

\noindent
Since both programs compute the same stream of pairs (value, weight), and the permutation $f([R_1:R_2]) = [R_2:R_1]$ preserves the uniform distribution, the two programs are equivalent.

\end{example*}

\section{Application: Assumed Parameters Filtering}
\label{sec:apf}

ProbZelus probabilistic models are state-space models that can involve two kinds of random variables.
\emph{State parameters} are represented by a stream of random variables which evolve over time depending on the previous values and the observations.
\emph{Constant parameters} are represented by a random variable whose value is progressively refined from the \emph{prior} distribution with each new observation.

\begin{example*}
Consider the example of \Cref{sec:example} where the boat is drifting at a constant speed~$\theta$.
We want to estimate both the moving position (state parameter), and the drift speed (constant parameter).
The motion model \zl{f} is now defined as follows (where the noise parameter \zl{st} is a global constant):

\begin{lstlisting}
let proba f(pre_x) = pre_x + theta where
  rec init theta = sample(gaussian(zeros, st))
  and theta = last theta
\end{lstlisting}
\end{example*}

\paragraph{Filtering}
To estimate state parameters, Sequential Monte Carlo~(SMC) inference algorithms rely on filtering techniques~\cite{doucet-smc-2006,chopin2020SMC}.
Filtering is an approximate method which consists of deliberately losing information on the current approximation to refocus future estimations on the most significant information.
These methods are particularly well suited to the reactive context where a system in interaction with its environment never stops and must execute with bounded resources.
All ProbZelus inference methods are SMC algorithms~\cite{rppl-short,ss-oopsla22,zlax_lctes22}.
Unfortunately, this loss of information is harmful for the estimation of constant parameters which do not change over time.

The most basic SMC algorithm, the \emph{particle filter}, approximates the posterior distribution by launching a set of independent executions, called \emph{particles}.
At each step, each particle returns a value associated to a score which measures the quality of the value w.r.t. the model.
To recenter the inference on the most significant particles, the inference runtime periodically resamples the set of particles according to their weights.
The most significant particles are then duplicated while the least interesting ones are dropped.

\begin{figure}
  \centering
  \begin{tabular}{@{}ccc@{}}
  \begin{subfigure}[b]{.32\textwidth}
  \begin{tikzpicture}
  \node[rotate=90] at (-2.3, 0) {\textsc{pf}}; 
  \node at (0, 0) {\includegraphics[width=\textwidth]{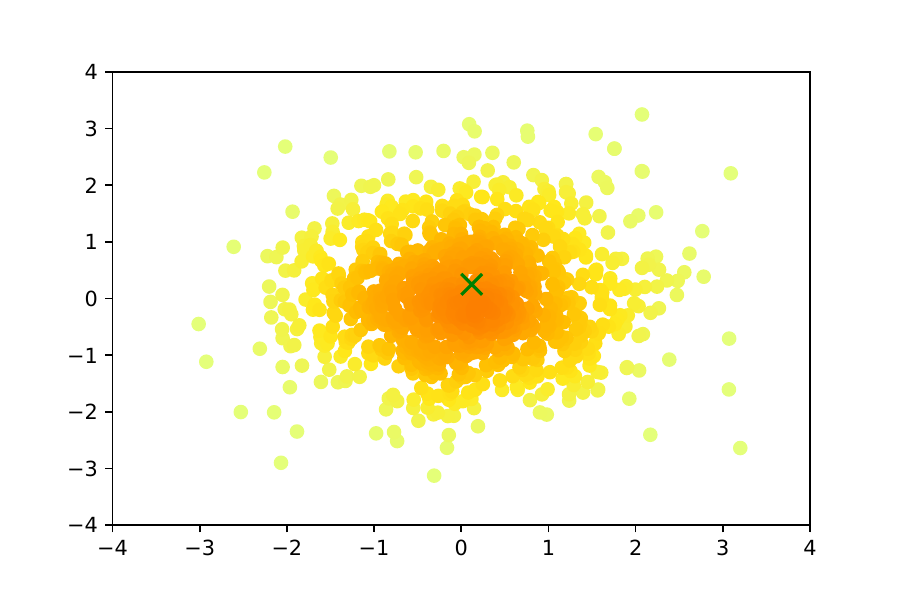}};
  \end{tikzpicture}
  \end{subfigure}
  &
  \begin{subfigure}[b]{.32\textwidth}

  \begin{tikzpicture}
  \node at (0, 0) {\includegraphics[width=\textwidth]{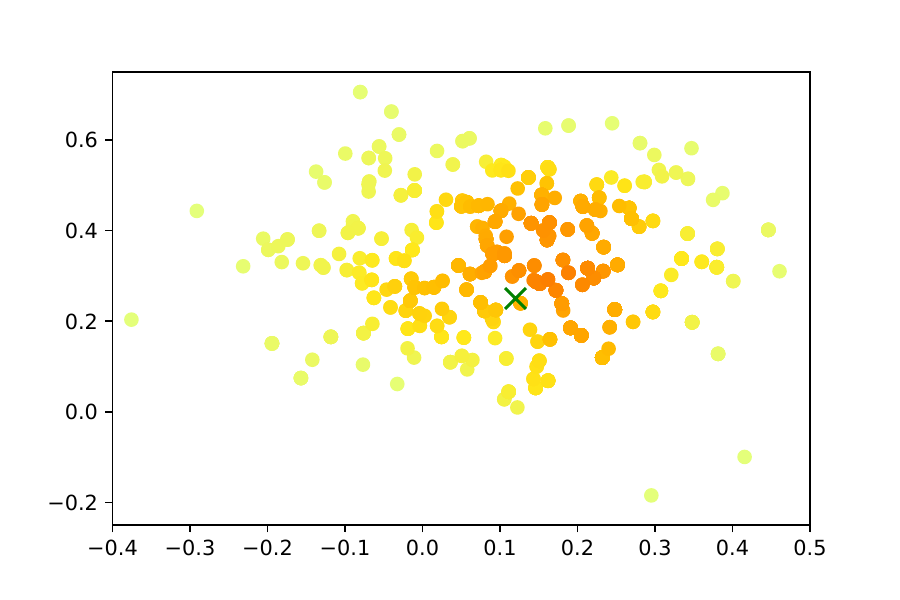}};
  \end{tikzpicture}
  \end{subfigure}
  &
  \begin{subfigure}[b]{.32\textwidth}
  \begin{tikzpicture}
  \node at (0, 0) {\includegraphics[width=\textwidth]{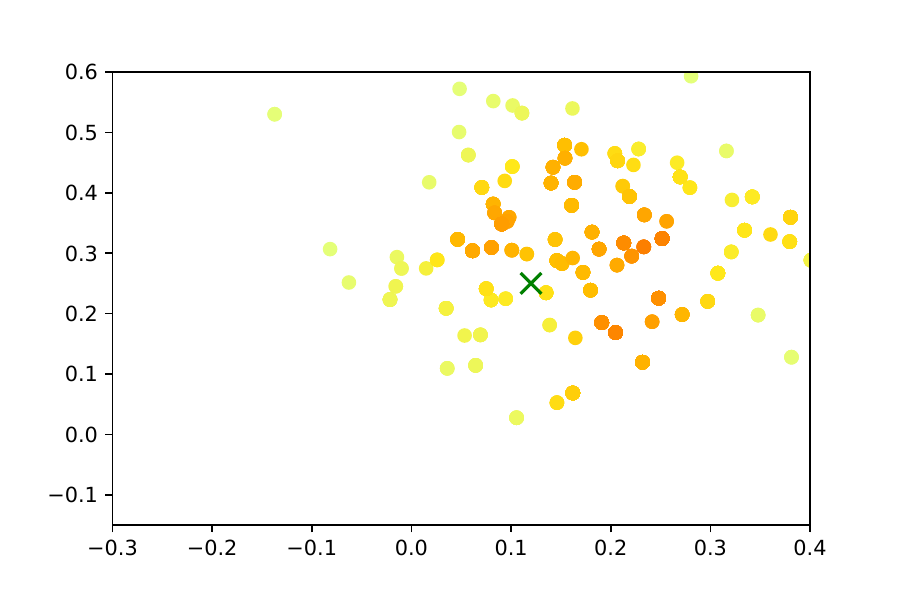}};
  \end{tikzpicture}
  \end{subfigure}
  \\[-1.5em]
  \begin{subfigure}[b]{.32\textwidth}
  \begin{tikzpicture}
  \node[rotate=90] at (-2.3, 0) {\textsc{apf}}; 
  \node at (0, 0) {\includegraphics[width=\textwidth]{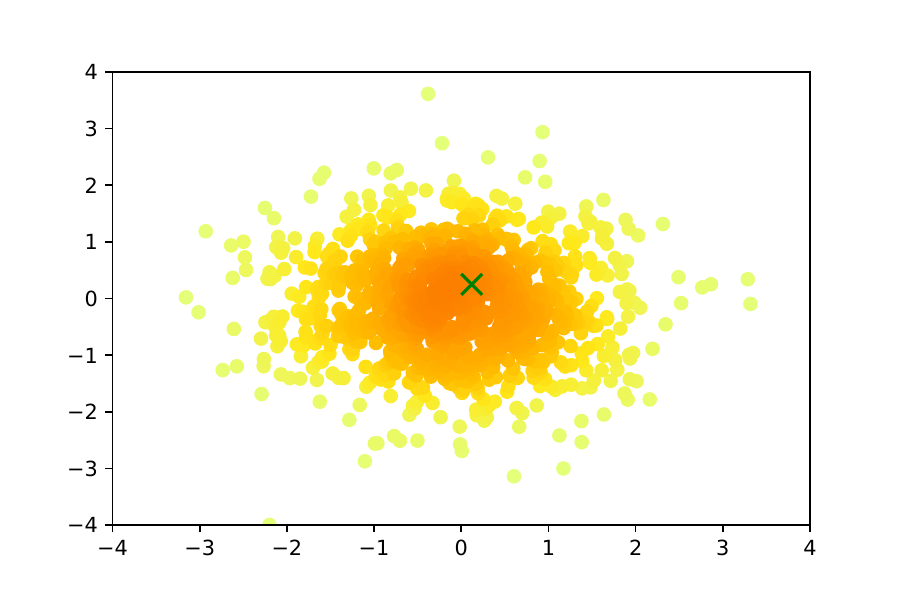}};
  \end{tikzpicture}
  \caption{$t=0$}
  \end{subfigure}
  &
  \begin{subfigure}[b]{.32\textwidth}
  \begin{tikzpicture}
  \node at (0, 0) {\includegraphics[width=\textwidth]{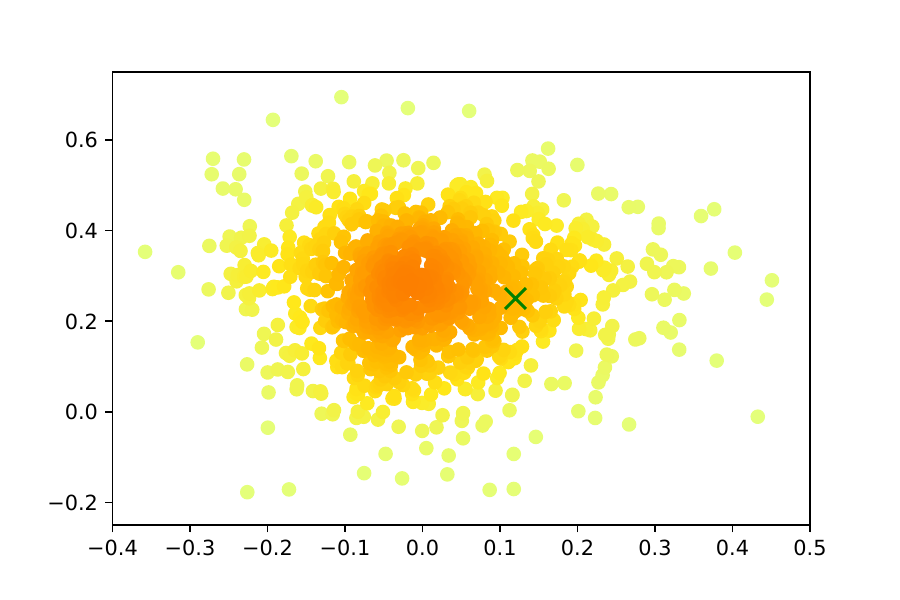}};
  \end{tikzpicture}
  \caption{$t=30$}

  \end{subfigure}
  &
  \begin{subfigure}[b]{.32\textwidth}
  \begin{tikzpicture}
  \node at (0, 0) {\includegraphics[width=\textwidth]{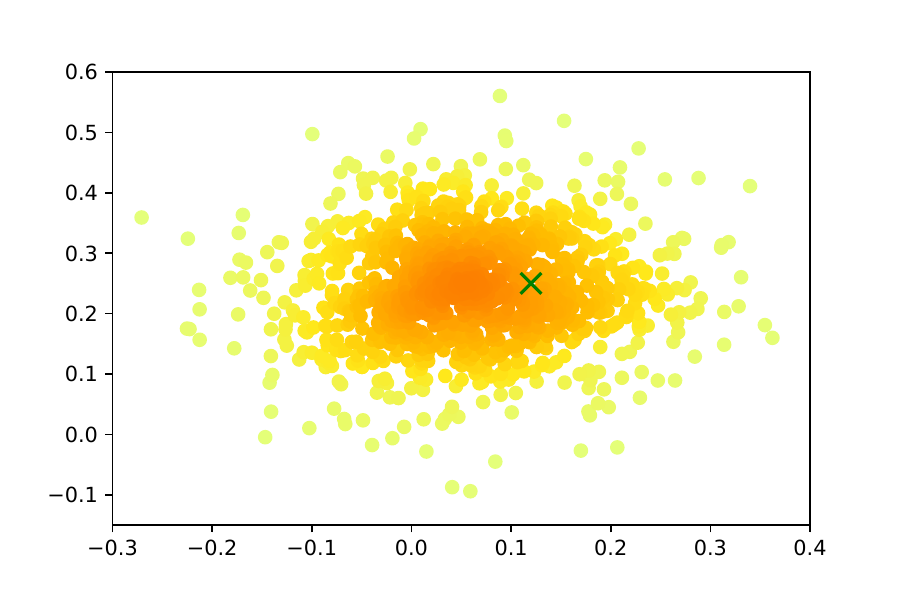}};
  \end{tikzpicture}
  \caption{$t=50$}
  \end{subfigure}
  \end{tabular}
  \caption{Estimates of the \texttt{theta} parameter of the radar example over time with a particle filter (\textsc{PF} at the top) and assumed parameter filter (\textsc{APF} at the bottom) .
   The true drift speed is indicated by a green cross.
   The color gradient represents the dot density.
   The scale shrinks over time.
   Results may differ in between runs.}
  \label{fig:impoverishment}
\end{figure}

Unfortunately, constant parameters are only sampled at the beginning of the execution of each particle.
After each resampling step, the duplicated particles share the same value for \zl{theta}.
The quantity of information used to estimate constant parameters thus strictly decreases over time until eventually, there is only one possible value left. 
The upper part of the \Cref{fig:impoverishment} graphically illustrates this phenomenon.

\paragraph{Assumed Parameter Filter}
To mitigate this issue, the \emph{Assumed Parameter Filter}~(APF)~\cite{ErolWLR17} split the inference into two steps: 1)~estimate the state parameters, and 2)~update the constant parameters.
The lower part of \Cref{fig:impoverishment} illustrates the results of APF on the estimation of the drift speed for our radar example. 

The APF algorithm assumes that constant parameters are well identified and that the model is written in a form which makes it possible to update the constant parameters given the state parameters.
The constant parameters must be an input of the model, and their prior distributions an input of a new inference operator \zl{APF.infer} (the APF algorithm is described in \Cref{app:apf:algo}).

\medskip

In this section, we present a program transformation which generates models that are exploitable by the APF algorithm.
First, a static analysis identifies the constant parameters and their prior distributions.
Then a compilation pass transforms these parameters into additional input of the model.
We use the relational semantics of \Cref{sec:rel} to prove the correctness of this transformation, i.e., the transformation preserves the ideal semantics of the program.

\begin{example*}
The compiled radar model for APF is the following:

\begin{lstlisting}
let f_prior = gaussian(zeros, id)
proba f_model(theta, pre_x) = pre_x + theta

let tracker_prior = f_prior
proba tracker_model(theta, y_obs) = x where
  rec init x = x_init
  and x = sample(gaussian(f(theta, last x), sx))
  and y = g(x)
  and () = observe(gaussian(y, sy), y_obs)

node main(y_obs) = msg where
    rec x_dist = APF.infer (tracker_model, tracker_prior, y_obs)
    and msg = controller(x_dist)
\end{lstlisting}
\end{example*}

\subsection{Static Analysis}
\label{sec:apf:typing}

\begin{figure}[t]
\begin{small}
\begin{mathpar}
\inferrule
    {\vdashl{C}{e}}
    {\vdashc{\apfenv, C}{\zldef{x}{e}}{\apfenv, C + \{x\}}}

\inferrule
    {\vdashc{C}{e}{\apftype}}
    {\vdashc{\apfenv, C}{\zlproba{f}{x}{e}}{\apfenv + \{f \is \apftype\}, C}}

\\

\inferrule
	{\vdashc{C}{e_1}{\apftype}}
	{\vdashc{C}{\zlpresent{e_1}{e_2}{e_3}}{\apftype}}

\inferrule
	{\vdashc{C}{e_2}{\apftype}}
	{\vdashc{C}{\zlreset{e_1}{e_2}}{\apftype}}

\inferrule
	{ }
	{\vdashc{C}{\zlapp{f_\theta}{e}}{\{\theta \is \zlfprior{f}\}}}

\inferrule
    {\vdashc{C}{e}{\apftype_e}\\
     \vdashleq{C}{E}{D}\\
     \vdashceq{C}{D}{E}{\apftype_E}}
    {\vdashc{C}{\zlwhere{e}{E}}{\apftype_e + \apftype_E}}

\inferrule
    {x \in D\\
     \vdashl{C}{e}}
    {\vdashceq{C}{D}{\zlisample{x}{}{e}}{\{x \is e\}}}

\inferrule
    {\vdashc{C}{e}{\apftype}}
    {\vdashceq{C}{D}{\zleq{x}{e}}{\apftype}}

\\
\inferrule
    { }
    {\vdashleq{C}{\zleq{x}{\zllast{x}}}{\{x\}}}

\inferrule
    {\vdashl{C}{e}}
    {\vdashleq{C}{\zleq{x}{e}}{\{x\}}}

\inferrule
    {x \in C}
    {\vdashl{C}{x}}

\inferrule
    {\vdashl{C + \dom{E}}{e}\\
     \vdashleq{C}{E}{\dom{E}}}
    {\vdashl{C}{\zlwhere{e}{E}}}
\end{mathpar}
\end{small}
\caption{Extract constant parameters and associated prior distributions (full type system in~\Cref{app:apf:typing}).}
\label{fig:apf_sa}
\end{figure}

The goal of the static analysis is to identify the constant parameters of each probabilistic node, i.e., initialized random variables~($\zlinit{x}{\zlsample{}{e}}$) that are also constant~($\zleq{x}{\zllast{x}}$).
For a program $\mathit{prog}$ the judgement $\vdashc{\emptyset, \emptyset}{\mathit{prog}}{\apfenv, C}$ builds the environment~$\apfenv$ which associates to each probabilistic node a type~$\apftype$ which maps constant parameters to their prior distributions. 
The environment~$C$ contains the global constant variables that can be used to define the prior distributions. 
The main type system is given in \Cref{fig:apf_sa}.

\paragraph{Constants}
The auxiliary judgement $\vdashleq{C}{E}{C'}$ identifies constant streams~$C'$ in the set of equations~$E$ given the constant variables~$C$.
A stream~$x$ is constant if it is always equal to its previous value~($\zleq{x}{\zllast{x}}$) or if it is defined by a constant expression.
The auxiliary judgement $\vdashl{C}{e}$ checks that an expression defines a constant stream.
An expression with a set of local declarations is constant if all the equations define constant streams.

\paragraph{Declarations}
A global declaration $\zldef{x}{e}$ typed in the environment $\apfenv, C$ adds the name~$x$ to the global constant set~$C$ if the expression~$e$ is constant.

A probabilistic node $\zlproba{f}{x}{e}$ is associated to the map~$\apftype$ computed by the judgement~$\vdashc{C}{e}{\apftype}$.

\paragraph{Expressions}
Typing an expression collects the constant parameters of the sub-expressions.
To simplify the analysis, we associate a unique instance name~$\theta$ to each function call and we assume that all variables and instances names are unique, e.g., $f(1) + f(2)$ becomes $f_{\theta_1}(1) + f_{\theta_2}(2)$.
The rule for~$\zlapp{f_{\theta}}{e}$ associates to~$\theta$ the prior distribution of the constant parameters of the body of~$f$: $\zlfprior{f}$ which is defined as a global variable by the compilation pass.
The rule for $\zlreset{e_1}{e_2}$ focuses only on the condition $e_2$ because~$e_1$ can be re-initialized and thus is not constant.
Similarly, the rule for $\zlpresent{e_1}{e_2}{e_3}$ focuses only on the condition~$e_1$.

\paragraph{Equations}
The typing of $\zlwhere{e}{E}$ identifies the set~$D$ of constant variables in~$E$ then types the equations with the judgement~\vdashceq{C}{D}{E}{\apftype} where~$C$ is the set of constant free variables in~$E$.
If a variable~$x$ is introduced by the equation $\zlinit{x}{\zlsample{}{e}}$ and is constant~($x \in D$), then~$x$ is a constant parameter and the result type maps $x$ to the distribution~$e$.

\begin{example*}
On our example, the variable \zl{theta} is identified as a constant parameter of the node \zl{f} and is propagated through the node \zl{tracker} that calls~\zl{f}. 
The final environment is:
$$\apfenv = \{\mathtt{f} \is \{ \mathtt{theta} \is \mathtt{gaussian(...)} \}, \mathtt{tracker} \is \{ \theta \is \zlfprior{\mathtt{f}} \} \}$$
\end{example*}

\subsection{Compilation}
\label{sec:apf:compilation}

\begin{figure}
\begin{small}
\begin{array}[t]{@{}lcl@{}}

\compile{\apfenv}{\zlproba{f}{x}{e}} &=& 
  \begin{array}[t]{@{}l@{\quad}l}
  \zldef{\zlfprior{f}}{\im{\apftype}}
  & \text{with } \apftype = \apfenv(f)\\
  \zlproba{\zlfmodel{f}}{\zlpair{\dom{\apftype}}{x}}{\compile{\apftype}{e}}
  \end{array}
\\[0.3em]

\compile{\apftype}{\zlwhere{e}{E}} &=& \zlwhere{\compile{\apftype}{e}}{\compile{\apftype}{E}}\\[0.5em]

\compile{\apftype}{\zlinit{x}{e}} &=& 
    \left \{
    \begin{array}{@{}ll@{}}
        \emptyset & \text{if $x \in \dom{\apftype}$}\\
        \zlinit{x}{\compile{\apftype}{e}} & \text{otherwise}
    \end{array}
    \right .
\\[1em]
\compile{\apftype}{\zleq{x}{e}} &=& 
    \left \{
    \begin{array}{@{}ll@{}}
        \emptyset & \text{if $x \in \dom{\apftype}$}\\
        \zleq{x}{\compile{\apftype}{e}} & \text{otherwise}
    \end{array}
    \right .
\\[1em]

\compile{\apftype}{\zlapp{f_\theta}{e}} &=& 
  \left \{
  \begin{array}{@{}ll@{}}
    \zlapp{f}{\compile{\apftype}{e}} &\text{if $f$ is deterministic}\\
    \zlapp{\zlfmodel{f}}{\theta, \compile{\apftype}{e}} &\text{if $\theta \in \dom{\apftype$}}\\
    \zlapp{\zlfmodel{f}}{\theta, \compile{\apftype}{e}} 
    \ \kwf{where} &\text{otherwise} \\ 
    \quad 
    \kwf{rec} \  \zlinit{\theta}{\zlsample{}{\zlfprior{f}}} \\
    \quad 
    \kwf{and} \  \zleq{\theta}{\zllast{\theta}} \\
  \end{array}
  \right .\\[3em]

\compile{\apftype}{\zlinfer{\zlapp{f}{e}}} &=& \zlapfinfer{}{}{\zlfmodel{f}}{\zlfprior{f}}{\compile{\apftype}{e}}

\end{array}
\end{small}
\caption{APF compilation (full definition in~\Cref{app:apf:compilation}).}
\label{fig:apf:compilation}
\end{figure}  

To run the APF algorithm, constant parameters must become additional inputs of the model.
The inference runtime can thus execute the model multiple times with different values of the constant parameters to update their distributions.
The compilation function is defined in \Cref{fig:apf:compilation} by induction on the syntax and relies on the result of the static analysis.
The compilation function~$\mathcal{C}$ is thus parameterized by the typing environment~$\apfenv$ for declarations and the type~$\apftype$ for expressions.

A model $\zlproba{f}{x}{e}$ such that $\apfenv(f) = \apftype$~(\textit{i.e.}, $\vdashc{C}{e}{\apftype}$) is compiled into two statements:
\begin{itemize}
\item $\zldef{\zlfprior{f}}{\im{\apftype}}$: the prior distribution of the constant parameters in $e$, and
\item $\zlproba{\zlfmodel{f}}{(\dom{\apftype}, x)}{\compile{\apftype}{e}}$: a new model that takes the constant parameters $\dom{\apftype}$ as additional arguments.
\end{itemize}

The compilation function of an expression~$\mathcal{C}_\apftype$ removes the definitions of the constant parameters $x \in \dom{\apftype}$.
The \zl{where}/\zl{rec} case effectively removes the constant parameters by keeping only the equations defining variables~$x \not\in \dom{\apftype}$.
The main difficulty is to handle constant parameters introduced by a function call $f_\theta(e)$.
\begin{itemize}
\item If the node is deterministic, there is no constant parameter.
\item If the constant parameters of the callee are also constant parameters of the caller, we have $\theta \in \dom{\apftype}$ and we just replace the call to $f_\theta$ with a call to $\zlfmodel{f}$ using the instance name for the constant parameters.
\item Otherwise, the constant parameters of the callee are not constant for the caller because the instance $f_\theta$ is used inside a \zl{reset}/\zl{every} or a \zl{present}/\zl{else}.
In this case, we redefine these parameters locally by sampling their prior distribution~$\zlfprior{f}$.
\end{itemize}

Finally, the call to $\zlinfer{\zlapp{f}{e}}$ is replaced by a call to $\zlapfinfer{}{}{\zlfmodel{f}}{\zlfprior{f}}{e}$.

\subsection{Correctness}

We use the relational semantics to prove the correctness of the APF compilation pass.
First, we prove that any probabilistic expression is equivalent to its compiled version computed in an environment which already contains the definition of the constant parameters.
The main theorem which relates $\zlinfer{\zlapp{f}{e}}$ and $\zlapfinfer{}{}{\zlfmodel{f}}{\zlfprior{f}}{e}$ then corresponds to the case $f_\theta(e)$.

\begin{definition}
\label{def:apf:sem}
For a model $f$ that is compiled into $\zlfprior{f}$ and $\zlfmodel{f}$,
the ideal semantics of \zl{APF.infer} externalizes the definition of the constant parameters.
{
$$
\begin{array}{c}
    \zlapfinfer{}{}{\zlfmodel{f}}{\zlfprior{f}}{e}
    \\\equiv\\
    \zlinfer{
        \zlwhere
          {\zlapp{\zlfmodel{f}}{\theta, e}}
          {\zlisample{\theta}{}{\zlfprior{f}}}}
\end{array}
$$}
\end{definition}

\paragraph{Notations}
In the following, $G^+$ is the context $G$ for global declarations augmented with the definition of $\attr{f}{model}$ and $\zlfprior{f}$ for all probabilistic functions $f$.
To simplify the presentation we use vectorized notations, e.g., $\sample{\vec{\mu}}{{R}_0}$ is a syntactic shortcut for the list $[\sample{\vec{\mu}[0]}{{R[0]}_0},\sample{\vec{\mu}[1]}{{R[1]}_0}, ...]$.

\begin{theorem}
\label{thm:apf}
For all probabilistic functions $f$ such that $\apfenv(f) = \apftype$,
for all expressions $e$ in the body of $f$ such that 
$\vdashc{C}{e}{\apftype_e}$, there is a permutation $R \to [R' : R^p]$ such that
\[
G, H, R \vdash e \Downarrow (s, w) 
\Longleftrightarrow 
G^+, H + H_f, R' \vdash \compile{\apftype}{e} \Downarrow (s, w).
\]
where $\vec{p} = \dom{\apftype_e}$ is the list of constant parameters in $e$, $\vec{\mu} = \im{\apftype_e}$ are the corresponding prior distributions and $H_f$ is the context that already contains the definitions of all the constant parameters in $f$ including $\vec{p}$, i.e., $H_f(\vec{p}) = \sample{\vec{\mu}}{R^p_0}$.
\end{theorem}

\begin{proof}
The proof is by induction on the expression  and the size of the context, i.e., the number of declarations before the expression.
For each induction case, we give the mapping between $R$ and $R' + R^p$, and show that a semantics derivation for $e$ in a context $G, H, R$ is equivalent to a semantics derivation for $\compile{\apftype}{e}$ in the context $G^+, H + H_f, R'$.
We focus on the most interesting cases, i.e., expressions altered by the compilation function.

To simplify the proof we prove the following lemma for equations in parallel with \Cref{thm:apf}.
\begin{lemma}
For all equations $E$ in the body of $f$ such that 
$\vdashc{C, D}{E}{\apftype_E}$, there is a permutation $R \to [R' : R^p]$ such that
\[
G, H, R \vdash E : W 
\Longleftrightarrow 
G^+, H + H_f, R' \vdash \compile{\apftype}{E} : W.
\]
\end{lemma}

\paragraph{Case 
$\zland
    {\zlisample{x}{}{d}}
    {\zleq{x}{\zllast{x}}}$ where $x \in \dom{\apftype}$.}

The permutation is $[R_x] \to []:[R_x]$ and with $v_x = \sample{\mu}{{R_x}_0}$

\[
\begin{footnotesize}
\begin{prooftree}
\infer0{
    G, H \vdash d \downarrow \mu \cdot s_\mu
}
\infer1{
  G, H + [x \is v_x], [R_x] \vdash \zlisample{x}{}{d} : 1
}
\hypo{
  G, H + [x \is v_x], [] \vdash \zleq{x}{\zllast{x}} : 1
}
\infer2{
    G, H + [x \is v_x], [R_x] 
    \vdash 
        \zland
            {\zlisample{x}{}{d}}
            {\zleq{x}{\zllast{x}}} 
    : 1
}
\end{prooftree}
\end{footnotesize}
\]

On the other hand, because  $x$ is a constant parameter, $\compile{\apftype}{d} = \emptyset$ and  $H_f(x) = v_x$.

\[
\begin{footnotesize}
\begin{prooftree}
\hypo{
    H_f(x) = v_x
}
\infer0{
    G^+, H + H_f, [] 
    \vdash \emptyset
    : 1
} 
\infer2{
    G^+, H + H_f, [] 
    \vdash 
        \compile{\apftype}{
            \zland
            {\zlisample{x}{}{d}}
            {\zleq{x}{\zllast{x}}}} : 1
}
\end{prooftree}
\end{footnotesize}
\]

\noindent
This results can then be generalized to arbitrary sets of equations where the two equations are not necessarily grouped together at the cost of an additional permutation.

\paragraph{Case $\zlapp{g_\theta}{e$}}
By induction we have the two permutations $R_e \to [R_e': R_e^p]$ and $R_g \to [R_g', R_g^p]$.
With $\zlproba{g}{x}{e_g}$ and $\vdashc{C}{e_g}{\apftype_g}$, we can apply the induction hypothesis on $e_g$ because there is no possible recursive call. 
The callee context for $e_g$ is thus strictly smaller than the caller context.
We also have $H_g = [\vec{p_g} \is \vec{{v_p}}]$ with $\vec{p_g} = \dom{\apftype_g}$, $\vec{\mu_g} = \im{\apftype_g}$, and $\vec{{v_p}} = \sample{\vec{\mu_g}}{{R^p_g}_0}$.

\[
\begin{footnotesize}
\begin{prooftree}
\hypo{
    G^+, H + H_f, R'_e
    \vdash \compile{\apftype}{e} 
    \Downarrow (s_e, w_e)
}
\infer1{
    G, H, R_e
    \vdash e 
    \Downarrow (s_e, w_e)
}    
\hypo{
    G^+, [x \is s_e] + [\vec{p_g} \is \vec{v_{p}}], R'_g
    \vdash \compile{\apftype_g}{e_g} 
    \Downarrow (s, w)
}
\infer1{
    G, [x \is s_e], R_g
    \vdash e_g 
    \Downarrow (s, w)
}
\infer2{
    G,  H, [R_e:R_g] \vdash \zlapp{g_\theta}{e} \Downarrow (s, w_e * w)
}
\end{prooftree}
\end{footnotesize}
\]

On the other end, by construction we have $G(\attr{g}{model}) = \zlproba{\attr{g}{model}}{\zlpair{\vec{p_g}}{x}}{\compile{\apftype_g}{e_g}}$ and there are two cases. If $\theta \in \dom{\apftype}$, then the constant parameters are already in the context and $H_f(\theta) = \vec{{v_p}}$. The permutation is $[R_e: R_g] \to [R_e':R_g']:[R_e^p:{R_g^p}]$ and we have:

\[
\begin{footnotesize}
\begin{prooftree}
\infer0{
    G^+, H + H_f, []
    \vdash \theta
    \Downarrow (\vec{v_p}, 1)
}
\hypo{
    G^+, H + H_f, R_e'
    \vdash \compile{\apftype}{e}
    \Downarrow (s_e, w_e)
}
\infer2{
    G^+, H + H_f, R_e'
    \vdash \zlpair{\theta}{\compile{\apftype}{e}}
    \Downarrow ((\vec{v_p}, s_e), w_e)
}
\hypo{
    G^+, [x \is s_e, \vec{p_g} \is \vec{v_p}], R_g' 
    \vdash \compile{\apftype_g}{e_g}
    \Downarrow (s, w)
}
\infer2{
    G^+, H + H_f, [R_e', R_g'] 
    \vdash \zlapp{\zlfmodel{g}}{\theta, \compile{\apftype}{e}}
    \Downarrow (s, w_e * w)
}
\infer1{
    G^+, H + H_f, [R_e', R_g'] 
    \vdash \compile{\apftype}{g_\theta(e)}
    \Downarrow (s, w_e * w)
}
\end{prooftree}
\end{footnotesize}
\]

\noindent
Finally if $\theta \notin \dom{\apftype}$, the constant parameters are not in the context and the compilation adds a defining equation for $\theta$.
The permutation is $[R_e: R_g] \to [R_e':R_g':R_g^p]:[R_e^p]$, and we have:

\[
\begin{footnotesize}
\begin{prooftree}
\hypo{...}
\infer1{
    G^+, H + H_f + [\theta \is \vec{v_p}], [R_e', R_g']
    \vdash \zlapp{\zlfmodel{g}}{\theta, \compile{\apftype}{e}}
    \Downarrow (s, w)
}
\hypo{
    G^+(\zlfprior{g}) = \vec{\mu_g}
}
\infer1{
    G^+, H + H_f, R_g^p
    \vdash \zlsample{}{\zlfprior{g}}
    \Downarrow (\vec{v_p}, 1)
}
\infer1{
    G^+, H + H_f + [\theta \is \vec{v_p}], R_g^p
    \vdash \zlisample{\theta}{}{\zlfprior{g}} : 1
}
\infer2{
    G^+, H + H_f, [R_e', R_g', R_g^p] 
    \vdash 
        \zlwhere
            {\zlapp{\zlfmodel{g}}{\theta, \compile{\apftype}{e}}}
            {\zlisample{\theta}{}{\zlfprior{g}}}
    \Downarrow (s, w)
}
\infer1{
    G^+, H + H_f, [R_e', R_g', R_g^p] 
    \vdash \compile{\apftype}{\zlapp{g_\theta}{e}}
    \Downarrow (s, w)
}
\end{prooftree}
\end{footnotesize}
\]

\end{proof}

We can now state and prove the correctness of the APF compilation pass.

\begin{theorem}[APF compilation]
For all probabilistic nodes $f$,
\[ 
    G, H \vdash \zlinfer{\zlapp{f}{e}} \downarrow d
    \Longleftrightarrow 
    G^+, H 
    \vdash 
    \zlapfinfer{}{}{\zlfmodel{f}}{\zlfprior{f}}{e}
    \downarrow d
\]
\end{theorem}

\begin{proof}
From \Cref{def:apf:sem}, we need to show that for all random streams $R$:
\[ 
    G, H, R \vdash \zlapp{f_\theta}{e} \Downarrow (s, w)
    \Longleftrightarrow 
    G^+, H + [\theta \is \vec{v_p}], R'
    \vdash  \zlapp{\zlfmodel{f}}{\theta, e} \Downarrow (s, w)
\]
with $G^+(\zlfprior{f}) = \vec{\mu}$ and $\vec{v_p} = \sample{\vec{\mu}}{R^p_0}$.
This corresponds to the case $\zlapp{f_\theta}{e}$ with $\theta \in \dom{\apftype}$.
\end{proof}
\section{Related work}

\paragraph{Probabilistic Semantics}
Measurable functions and kernels have been used to define the semantics of first-order probabilistic programs as probability distribution transformers~\cite{kozen81}. 
Probabilistic Coherent Spaces is a generalization of this idea to higher-order types but for discrete probability~\cite{DanosE11}. 
This setting has been extended to continuous distributions with models based on positive cones~\cite{EhrhardPT18,DahlqvistK20}, a variation on Banach spaces with positive scalars~\cite{Selinger04}. 
To interpret the sampling operation, cones have to be equipped with a measurability structure such that measures and integration can be defined for any types~\cite{EhrhardGeoffroy23}.

Quasi-Borel spaces~(QBS) are another alternative to classic measurable spaces to define the semantics of higher-order probabilistic programs with conditioning~\cite{StatonYWHK16, HeunenKSY17}.
A probabilistic expression is interpreted as a quasi-Borel measure, i.e., an equivalence class of pairs $[\alpha, \mu]$ where $\mu$ is a measure over $\mathbb{R}$, and $\alpha$ is a measurable function from $\mathbb{R}$ to values.
Intuitively, the corresponding distribution is obtained as the pushforward of $\mu$ along $\alpha$.
Recent work extends this formalism to capture lazy data structures and streams in a functional probabilistic language~\cite{DashKPS23}.

Our density-based semantics relies on a similar representation: probabilistic expressions are interpreted by pushingforward a uniform measure over $[0, 1]^n$ along a measurable function.
The main difference is that we focus on a domain specific dataflow synchronous language.
The set of random variables can be computed statically, and integration is entirely deferred to the \zl{infer} operator.
Importantly, we recover the fact that equations sets can be interpreted in any order, a key property for dataflow synchronous languages.

\paragraph{Program Equivalence}

Probabilistic bisimulation has been introduced for testing equivalence of discrete probabilistic systems~\cite{LarsenS89DBisimulation} and generalized to study Labelled Markov Process (continuous systems)~\cite{DesharnaisEP02}. 
Following~\cite{LagoG19} which defines a notion of bisimulation for a (higher-order) probabilistic lambda calculus, we could define a notion of probabilistic bisimulation for the co-iterative semantics. 
One challenge is to adapt the bisimulation to reactive state machines with explicit state. 

Probabilistic coupling, is an alternative classic proof technique for probabilistic programs equivalence~\cite{Hsu2017Coupling}. 
A coupling describes correlated executions by associating pairs of samples. \Cref{prop:equiv} that we use to prove the correctness of the APF compilation pass is an instance of probabilistic coupling where the relation is made explicit through a measurable function.

\paragraph{Static analysis and compilation}
The static analysis presented in \Cref{sec:apf:typing} is inspired by static analyses designed for dataflow synchronous languages, in particular the initialization analysis~\cite{lucy:sttt04} which guarantees that all streams are well defined at the first time step, and the typing of Zelus' static arguments which checks that some value can be statically computed at compile time~\cite{BourkeCCPPP17}.

The compilation pass presented in \Cref{sec:apf:compilation} turns constant parameters into additional arguments of the model.
An alternative, closer to our density-based semantics, would be to externalize all random variables.
This would give a lot of flexibility to the inference runtime to apply different inference strategies to different sets of variables.
This approach is also reminiscent of the compilation of Zelus' hybrid models~\cite{lucy:lctes11}.

\paragraph{Inference}
ProbZelus inference runtime relies on semi-symbolic sequential Monte Carlo samplers~\cite{lunden17, rppl-short, ss-oopsla22}.
The posterior distribution is approximated by a set of independent executions, the particles.
Each particle tries to compute a closed form solution using symbolic computation, and only samples concrete values when symbolic computations fails.
For some models, symbolic computations can be used to compute the distribution of constant parameters, and ProbZelus default inference algorithms thus outperforms APF.
Unfortunately symbolic computations are only possible with conjugate random variables and affine transformations.
APF is thus more robust for arbitrary models.
\section{Conclusion}

In this paper we proposed two semantics for a reactive probabilistic programming languages: a density-based co-iterative semantics and a density-based relational semantics.
Both semantics are schedule agnostic, i.e., sets of of mutually recursive equations can be interpreted in arbitrary order, a key property of synchronous dataflow languages.
The relational semantics directly manipulates streams which can significantly simplify program equivalence reasoning for probabilistic expressions.
We then defined a program transformation required to run an optimized inference algorithm for state-space models with constant parameters and used the relational semantics to prove the correctness of the transformation.

\paragraph{Acknowledgements}
The authors are grateful to the following people for many discussions and helpful feedback: Timothy Bourke, Grégoire Bussone, Thomas Ehrhard, Patrick Hoscheit, Adrien Guatto, Paul Jeanmaire, Basile Pesin, and Marc Pouzet.
This work was supported by the project ReaLiSe, Émergence Ville de Paris 2021-2025.

\bibliography{biblio}

\newpage
\appendix
\section{Semantics}

\subsection{Density-based co-iterative semantics}

\begin{figure}
\begin{small}
\[
\begin{array}{l@{\,\,}l@{\,\,\,}l}
\dsemi{e}_\gamma &=& \semi{e}_\gamma, 0\\
\dsems{e}_\gamma(m, []) &=& 
    \letin{m', v = \sems{e}_\gamma(m)} 
    m', v, 1
    \quad \text{if $e$ is deterministic}
\\\\
\dsemi{\zlsample{}{e}}_{\gamma} &=& \letin{m = \semi{e}_\gamma} m, 1\\
\dsems{\zlsample{}{e}}_{\gamma}(m, [r]) &=&
    \letin{m', \mu  = \sems{e}_\gamma(m)} 
    m', \sample{\mu}{r}, 1
\\\\
\dsemi{\zlfactor{e}}_{\gamma} &=&  \letin{m = \semi{e}_\gamma} m, 0\\
\dsems{\zlfactor{e}}_{\gamma}(m, []) &=& 
    \letin{m', v = \sems{e}_\gamma(m)} 
    m', (), v
\\\\
\dsemi{f(e)}_{\gamma} &=& 
    \begin{stack}
        \letin{m_f, p_f = \gamma(\attr{f}{init})}\\
        \letin{m_e, p_e = \dsemi{e}_\gamma}\\
        (m_f, m_e), p_f + p_e 
    \end{stack}\\
\dsems{f(e)}_{\gamma}((m_f, m_e), [r_f:r_e]) &=&
    \begin{stack}
        \letin{m_e', v_e, w_e = \dsems{e}_\gamma(m_e, r_e)}\\
        \letin{m_f', v, w_f = \gamma(\attr{f}{step})(v_e, m_f, r_f)}\\
        (m_f', m_e'), v, w_e * w_f
    \end{stack}
\\\\
\dsemi{\zlpresent{e}{e_1}{e_2}}_{\gamma} &=&
    \begin{stack}
      \letin{m_1, p_1 = \dsemi{e_1}_\gamma}\\
      \letin{m_2, p_2 = \dsemi{e_2}_\gamma}\\
      (\semi{e}_\gamma, m_1, m_2), p_1 + p_2
    \end{stack}\\
\dsems{\zlpresent{e}{e_1}{e_2}}_{\gamma}((m, m_1, m_2), [r_1:r_2]) &=&
    \begin{stack}
      \letin{m', v = \sems{e}_\gamma(m)}\\
      \ifthen{v}
        {\letin{(m_1', v_1, w) = \dsems{e_1}_\gamma(m_1, r_1)\\}
          (m', m_1', m_2), v_1, w}
        {\letin{(m_2', v_2, w) = \dsems{e_2}_\gamma(m_2, r_2)\\}
          (m', m_1, m_2'), v_2, w}
    \end{stack}
\\\\
\dsemi{\zlreset{e_1}{e_2}}_{\gamma} &=&
    \letin{m_1, p = \dsemi{e_1}_\gamma} (m_1, m_1, \semi{e_2}_\gamma), p\\
\dsems{\zlreset{e_1}{e_2}}_{\gamma}((m_0, m_1, m_2), w, r) &=&
    \begin{stack}
      \letin{m_2', v_2 = \sems{e_2}_\gamma(m_2)}\\
      \letin{m_1', v_1, w = \dsems{e_1}_\gamma
       (\mit{if} v_2 \mit{then} m_0 \mit{else} m_1, r)}\\
     (m_0, m_1', m_2'),v_1, w
    \end{stack}
\end{array}
\]
\end{small}
\caption{Density-based co-iterative semantics for ProbZelus probabilistic expressions.}
\label{sem:coit:expr:full}
\end{figure}

\begin{figure}
  \begin{small}
  \[
  \begin{array}{lll}
\dsemi{\zleq{x}{e}}_{\gamma} &=& \dsemi{e}_\gamma\\
\dsems{\zleq{x}{e}}_{\gamma}(m, r) &=&
    \letin{m',v, w = \dsems{e}_\gamma(m, r)}
     m',  [x \is v], w
\\\\
\dsemi{\zlinit{x}{e}}_{\gamma} &=& 
    \letin{m_0, p = \dsemi{e}_\gamma} 
    (\nil, m_0), p\\
\dsems{\zlinit{x}{e}}_{\gamma}((\nil, m_0), r) &=&
    \letin{m', i, w = \dsems{e}_\gamma(m_0, r)} 
    (\gamma(x), m_0), [\attr{x}{last} \is i], w\\
\dsems{\zlinit{x}{e}}_{\gamma}((v, m_0), r) &=&
    (\gamma(x), m_0), [\attr{x}{last} \is v], 1
\\\\
\dsemi{\zland{E_1}{E_2}}_{\gamma} &=& 
    \begin{stack}
        \letin{M_1, p_1 = \dsemi{E_1}_\gamma}\\
        \letin{M_2, p_2 = \dsemi{E_2}_\gamma}\\
        (M_1, M_2), p_1 + p_2
    \end{stack}\\
\dsems{\zland{E_1}{E_2}}_{\gamma}((M_1, M_2), [r_1: r_2]) &=&
    \begin{stack}
        \letin{M_1', \rho_1, w_1 = \dsems{E_1}_\gamma(M_1, r_1)}\\
        \letin{M_2', \rho_2, w_2 = \dsems{E_2}_{\gamma}(M_2, r_2)}\\
        (M_1', M_2'), \rho_1 + \rho_2, w_1 * w_2
    \end{stack}
\\\\
\dsemi{\zlwhere{e}{E}}_{\gamma} &=& 
    \begin{stack}
        \letin{m, p_e = \dsemi{e}_\gamma}\\
        \letin{M, p_E = \dsemi{E}_\gamma}\\ 
        (m, M), p_e + p_E
    \end{stack}\\
\dsems{\zlwhere{e}{E}}_{\gamma}((m, M), [r_e : r_E]) &=& 
    \begin{stack}
        \letin{F(\rho) = \left(
            \letin{M', \rho, w = \dsem{E}_{\gamma + \rho}(M, r_E)} \rho \right)}\\   
        \letin{\rho = \fix{F}}\\
        \letin{M', \rho, W = \dsems{E}_{\gamma + \rho}(M, r_E)}\\
        \letin{m', v, w = \dsems{e}_{\gamma'}(m, r_e)}\\
        (m', M'), v, w * W
    \end{stack}
\end{array}
\]
\end{small}
\caption{Density-based co-iterative semantics for ProbZelus equations.}
\label{sem:coit:eq:full}
\end{figure}

The density-based co-iterative semantics is presented in \Cref{sec:coit}.
\Cref{sem:coit:expr:full,sem:coit:eq:full} presents the full semantics for expressions and equations. The additional rules are for \zl{present} and \zl{reset}.

The initialization of $\zlpresent{e}{e_1}{e_2}$ allocates memory for~$e$, $e_1$ and~$e_2$ and count the number of random variables in~$e_1$ and~$e_2$~($e$ is deterministic and does not have any random variable).
The step function first executes~$e$ and depending on its value executes~$e_1$ or~$e_2$.
The initialization of~$\zlreset{e_1}{e_2}$ duplicates the memory needed to execute~$e_1$. That way, in the step function, only the second copy is updated by the transition and if~$e_1$ is reset, the execution restarts from the initial memory state. 

\subsection{Density-based relational semantics}

\paragraph{Stream functions}
The density-based relational semantics is presented in \Cref{sec:rel}.
The definition of this semantics relies on a few stream functions.

\[
\begin{array}{rcl}
\tlOp &:& \Stream{A} \to \Stream{A}\\
\tl{(a \cdot as)} &=& as\\
\\
\mapOp &:& (A \to B) \to (\Stream{A} \to \Stream{B})\\
\map{f}{(a \cdot as)} &=& f(a) \cdot (\map{f}{as})\\
\\
\mergeOp &:& \Stream{\mathbb{B}} \to \Stream{A} \to \Stream{A} \to \Stream{A}\\
\merge{(T \cdot cs)}{(a \cdot as)}{bs} &=& a \cdot (\merge{cs}{as}{bs})\\
\merge{(F \cdot cs)}{as}{(b \cdot bs)} &=& b \cdot (\merge{cs}{as}{bs})\\
\\
\whenOp &:& \Stream{\mathbb{A}} \to \Stream{\mathbb{B}} \to \Stream{A} \\
\when{(a \cdot as)}{(T \cdot cs)} &=& a \cdot (\when{as}{cs})\\
\when{(a \cdot as)}{(F \cdot cs)} &=& \when{as}{cs}\\
\\
\slicerOp &:& \Stream{(\Stream{A})} \to \Stream{\mathbb{B}} \to \Stream{A}\\
\slicer{((a \cdot as) \cdot bs \cdot ss)}{(F \cdot cs)} &=&
    a \cdot (\slicer{(as \cdot ss)}{cs})\\
\slicer{(as \cdot (b \cdot bs) \cdot ss)}{(T \cdot cs)} &=&
    b \cdot (\slicer{(bs \cdot ss)}{cs})\\
\end{array}
\]

\noindent
The function $\tlOp$ drops the first element of a stream~($\tlOp^n$ drops the~$n$ first elements).
The function $\map{f}{s}$ applies~$f$ to each element of the stream~$s$.
$\merge{cs}{as}{bs}$ merges the streams~$as$ and~$bs$ according to the condition~$cs$.
$\when{as}{cs}$ keeps the values of~$as$ only when the condition~$cs$ is true.
The function $\slicer{ss}{cs}$ is used to define the semantics of $\zlreset{e_1}{e_2}$. The first argument~$ss$ is a stream of streams where each stream represents~$e_1$ restarted at each time step, and~$cs$ the the reset condition.
When the condition is false, the first value of the first stream of~$ss$ is returned and the second stream of~$ss$ is discarded.
It means that we progress by one step in~$e_1$ and the stream representing~$e_1$ restarted at the current iteration is not useful since the expression was not reset.
When the condition is true, the first stream of~$ss$ which represents the current state of~$e_1$ is discarded and the execution restarts with the first value of the second stream of~$ss$ which represents~$e_1$ restarted at the current step.

\paragraph{Environment}
An \emph{environment} $H$ is a map from variable names to streams of values, for any bound variable $x \in \dom{H}$, $H(x): \Stream{A}$. 
When the context is clear, we write $f \ H$ for $\map{f}{H}$, e.g., for all $x \in \dom{H}$, $(\tl{H})(x) = \tl{(H(x))}$.

\begin{figure}
\begin{small}
\begin{mathpar}
\inferrule
    {}
    {G, H \vdash c \downarrow c}
\and
\inferrule
    {}
    {G, H \vdash x \downarrow H(x)}
\and
\inferrule
    {x \notin H}
    {G, H \vdash x \downarrow G(x)}
\and
\inferrule
    {G, H \vdash e_1 \downarrow s_1\\ 
     G, H \vdash e_2 \downarrow s_2}
    {G, H \vdash \zlpair{e_1}{e_2} \downarrow (s_1, s_2)}
\and
\inferrule
    {G, H \vdash e \downarrow s}
    {G, H \vdash \zlop{e} \downarrow \op(s)}
\and
\inferrule
    {H(\attr{x}{last}) = s \\ }
    {G, H \vdash \zllast{x} \downarrow s}
\and
\inferrule
    {G, H \vdash e \downarrow s_e\\
     G(f) = \zlnode{f}{x}{e_f} \\
     G, [x \is s_e] \vdash e_f \downarrow s}
    {G, H \vdash f(e) \downarrow s}
\and
\inferrule
    {G, H + H_E \vdash E \\ 
        G, H + H_E \vdash e \downarrow s}
    {G, H \vdash \zlwhere{e}{E} \downarrow s}
\and
\inferrule
    {G, H \vdash e \downarrow s_c \\ 
        G, (\when{H}{s_c})\vdash e_1 \downarrow s_1 \\
        G,  (\when{H}{\mathit{not}\ s_c}) \vdash e_2 \downarrow s_2}
    {G, H \vdash \zlpresent{e}{e_1}{e_2} \downarrow \merge{s_c}{s_1}{s_2}}
\and
\inferrule
    {\left[ G, (\tl^n \ H) \vdash e_1 \downarrow s_n \right]_{n \in \mathbb{N}} \\
     G, H \vdash e_2 \downarrow s_c}
    {G, H \vdash \zlreset{e_1}{e_2} \downarrow \slicer{(s_0 \cdot s_0 \cdot s_1 \cdot s_2 \cdot ...)}{s_c}}
\and
\inferrule
    {G, H \vdash e \downarrow H(x)}
    {G, H \vdash \zleq{x}{e}}
\and
\inferrule
    {G, H  \vdash e \downarrow v_i \cdot s_i \\ 
        H(\attr{x}{last}) = v_i \cdot H(x)}
    {G, H  \vdash \zlinit{x}{e}}
\and
\inferrule
    {G, H \vdash E_1 \\
     G, H \vdash E_2}
    {G, H \vdash \zland{E_1}{E_2}}
\end{mathpar}
\end{small}
\caption{Deterministic relational semantics.}
\label{fig:sem:rel:det:full}
\end{figure}

\begin{figure}
  \begin{small}
  \begin{mathpar}
  \inferrule
      {}
      {G, H, [] \vdash c \Downarrow (c, 1)}
  \and
  \inferrule
      {}
      {G, H, [] \vdash x \Downarrow (H(x), 1)}
  \and
  \inferrule
      {x \not\in H}
      {G, H, [] \vdash x \Downarrow (G(x), 1)}
  \and
  \inferrule
      {G, H, R_e \vdash e \downarrow (s_e, w_e)\\
       G(f) = \zlproba{f}{x}{e_f} \\ 
       G, [x \is s_e], R_f \vdash e_f \Downarrow (s, w)}
      {G, H, [R_e:R_f] \vdash \zlapp{f}{e} \Downarrow (s, w * w_e)}
  \and
  \inferrule
      {G, H + H_E, R_E \vdash E : w_E \\ 
          G, H + H_E, R_e \vdash e \Downarrow (s, w)}
      {G, H, [R_e:R_E] \vdash \zlwhere{e}{E} \Downarrow (s, w * w_E)}
  \and
  \inferrule
      {G, H, R_e \vdash e \downarrow s_c \\ 
          G, (\when{H, R_1}{s_c})  \vdash e_1 \Downarrow sw_1 \\
          G, (\when{H, R_2}{\mathit{not}\ s_c}) \vdash e_2 \Downarrow sw_2}
      {G, H, [R_1:R_2] \vdash \zlpresent{e}{e_1}{e_2} \Downarrow \merge{s_c}{sw_1}{sw_2}}
  \and
  \inferrule
      {\left[ G, (\tl^n\ H, R_1) \vdash e_1 \Downarrow sw_n \right]_{n \in \mathbb{N}} \\
       G, H, R_2 \vdash e_2 \downarrow s_c\\
       (s, w) = \slicer{(sw_0 \cdot sw_0 \cdot sw_1 \cdot sw_2 \cdot ...)}{s_c}}
      {G, H, [R_1:R_2] \vdash \zlreset{e_1}{e_2} \Downarrow (s, w)}
  \and
  \inferrule
      {G, H \vdash e \downarrow s_\mu}
      {G, H, [R] \vdash \zlsample{\alpha}{e} \Downarrow (\sample{s_\mu}{R}, 1)}
  \and
  \inferrule
      {G, H \vdash e \downarrow w}
      {G, H, [] \vdash \zlfactor{e} \Downarrow ((), w)}
  \and
  \inferrule
      {G, H, R \vdash e \Downarrow (H(x), w)}
      {G, H, R \vdash \zleq{x}{e} : w}
  \and
  \inferrule
      {G, H, R \vdash e \Downarrow (i \cdot s, w_i \cdot w) \\ 
          H(\attr{x}{last}) = i \cdot H(x)}
      {G, H, R \vdash \zlinit{x}{e} : w_i \cdot 1}
  \and
  \inferrule
      {G, H, R_1 \vdash E_1 : w_1 \\
       G, H, R_2 \vdash E_2 : w_2}
      {G, H, [R_1:R_2] \vdash \zland{E_1}{E_2} : w_1 * w_2}
  \and
  \inferrule
      { p = \RV(e)\\
        \left[G, H, R \vdash e \Downarrow (s, w) \qquad \overline{w} = \Pi\ w \right]_{R \in ([0,1]^\omega)^{p}}}
      {G, H \vdash \zlinfer{e} \downarrow \integ{p}{\overline{w}}{s}}
  \end{mathpar}
  \end{small}
  \caption{Probabilistic relational semantics.}
  \label{fig:sem:rel:prob:full}
  \end{figure}

\paragraph{Relational semantics}
The full deterministic and probabilistic density-based relational semantics including the rules for \zl{present} and \zl{reset} are given in \Cref{fig:sem:rel:det:full,fig:sem:rel:prob:full}.
The semantics of $\zlpresent{e}{e_1}{e_2}$ uses the $\whenOp$ function on the environment~$H$ such that the execution of~$e_1$ and~$e_2$ respectively progress only when the condition is true or false. Then the value of these two streams are merged using the $\mergeOp$ function.
The semantics of $\zlreset{e_1}{e_2}$ is based on the $\slicerOp$ function. $\left[ G, (\tl^n \ H) \vdash e_1 \downarrow s_n \right]_{n \in \mathbb{N}}$ represents the stream of streams where~$s_n$ is the stream of values computed by~$e_1$ restarted at time step~$n$. 
In the slicer, the stream~$s_0$ is duplicated because $\zlreset{e_1}{e_2}$ returns the same value whether or not $e_2$ is true at the initial step.

\section{Assumed Parameters Filtering}
\label{app:apf}

\subsection{Algorithm}
\label{app:apf:algo}

The inference methods proposed by ProbZelus~\cite{rppl-short,zlax_lctes22,ss-oopsla22} belong to the family of SMC algorithms.
These methods rely on a set of independent simulations, called \emph{particles}.
Each particle returns an output value associated with a score.
The score represents the quality of the simulation.
A large number of particles makes it possible to approximate the desired distribution.

More concretely, the \zl{sample(d)} construct randomly draws a value from the \zl{d} distribution, and the \zl{factor(x)} construct multiplies the current score of the particle by~\zl{x}.
At each instant, the \zl{infer} operator accumulates the values calculated by each particle weighted by their scores to approximate the posterior distribution.

If the model calls on the operator \zl{sample} at each instant, for example to estimate the position of the boat in the radar example~(\Cref{sec:example}), the previous method implements a simple random walk for each particle.
As time progresses, it becomes increasingly unlikely that one of the random walks will coincide with the stream of observations.
The score associated with each particle quickly goes down towards~$0$.

\begin{algorithm}[t]
\KwData{probablisitic model $\ttf{model}$, observation $y_t$, and previous result $\mu_{t-1}$.}
\KwResult{$\mu_t$ an approximation of the distribution of $p_t$.}
\BlankLine
  \For{each particle $i = 1$ \KwTo $N$}
  { 
    $p_{t-1}^i = \zlsample{\mu_{t-1}}$\\
    $p_t^i, w_t^i = \ttf{model(} y_t \mid p_{t-1}^i \ttf{)}$\\
  }
  $\mu_t = \mathcal{M}(\{w_t^i, p_t^i\}_{1 \leq i \leq N})$\\
  \Return $\mu_t$\\
\BlankLine
\BlankLine
\caption{Particle Filter.}
\label{algo:pf}
\end{algorithm}

To solve this issue, sequential Monte Carlo methods~(SMC) add a filtering step.
\Cref{algo:pf} describes the execution of one instant for a particle filter, the most basic SMC algorithm.
At each instant $t$, a particle $1 \leq i \leq N$ corresponds to a possible value of the parameters~(\textit{i.e.}, random variables) $p_{t}^i$ of the model.
We begin by sampling a new set of particles in the distribution obtained at the previous step.
The most probable particles are thus duplicated and the less probable ones are eliminated.
This refocuses the inference around the most significant information while maintaining the same number of particles throughout the execution.
Knowing the previous state~$p_{t-1}^i$, each particle then executes a step of the model to obtain a sample of the parameters~$p_t^i$ associated with a score~$w_t^i$.
At the end of the instant, we construct a distribution $\mu_t$ where each particle is associated with its score.
$\mathcal{M}(\{w_t^i, p_t^i\}_{1 \leq i \leq N})$ is a multinomial distribution, where the value $p_t^i$ is associated with the probability~$w^i_t / \sum_{i=1}^N w^i_t$.

Unfortunately, this approach generates a loss of information for the estimation of constant parameters.
On our radar example, at the first instant, each particle draws a random value for the parameter \zl{theta}.
At each instant, the duplicated particles share the same value for \zl{theta}.
The quantity of information useful for estimating \zl{theta} therefore decreases with each new filtering and, after a certain time, only one possible value remains.

Rather than sampling at the start of execution a set of values for the constant parameters that will impoverish with each filtering, in the APF algorithm, each particle computes a symbolic distribution of constant parameters.
At runtime, the inference then alternates between a sampling pass to estimate the state parameters, and an optimization pass which updates the constant parameters.
This avoids impoverishment for the estimation of the constant parameters.

\begin{algorithm}[t]
\KwData{probabilistic model $\ttf{model}$, observation $y_t$, and previous result $\mu_{t-1}$.}
\KwResult{$\mu_t$ an approximation of the distributions of state parameter $x_t$ and constant parameter $\theta$.}
\BlankLine
\For{each particle $i = 1$ \KwTo $N$}
{ 
  $x_{t-1}^i,\Theta_{t-1}^i  = \zlsample{\mu_{t-1}}$\\
  $\theta^i = \zlsample{\Theta_{t-1}^i}$\\
  $x_t^i, w_t^i = \ttf{model(} y_t \mid \theta^i, x_{t-1}^i \ttf{)}$\\
  $\Theta_{t}^i = \mathit{Udpate}(\Theta_{t-1}^i, \fun{\theta} \ttf{model(} y_t \mid \theta, x_{t-1}^i, x_t^i\ttf{)})$\\  
}
$\mu_t = \mathcal{M}(\{w_t^i, (x_t^i, \Theta_{t}^i)\}_{1 \leq i \leq N})$\\
\Return $\mu_t$\\
\BlankLine
\BlankLine
\caption{\emph{Assumed Parameter Filter}~\cite{ErolWLR17}.}
\label{algo:apf}
\end{algorithm}

More formally, \Cref{algo:apf} describes the execution of one step of APF.
At each instant~$t$, a particle $1 \leq i \leq N$ corresponds to a possible value of the state parameters $x_{t}^i$ and a distribution of constant parameters $\Theta_{t}^i$.
As for the particle filter, we begin by sampling a set of particles in the distribution obtained at the previous instant.
We then sample a value $\theta^i$ in $\Theta_{t-1}^i$.
Knowing the value of the constant parameters~$\theta^i$ and the previous state~$x_{t-1}^i$, we can execute a step of the model to obtain a sample of the state parameters $x^i_t$ associated with a score $w^i_t$.
We can then update $\Theta_{t}^i$ by exploring the other possible values for $\theta$ knowing that the particle has chosen the transition $x_{t-1}^i \rightarrow x_t^i$.
At the end of the instant, we construct a distribution $\mu_t$ where each particle is associated with its score. $\mathcal{M}(\{w_t^i, (x_t^i, \Theta_{t}^i)\}_{1 \leq i \leq N})$ is a multinomial distribution where the pair of values~$(x_t^i, \Theta_{t}^i)$ is associated to the probability~$w^i_t / \sum_{i=1}^N w^i_t$.

\subsection{Static Analysis}
\label{app:apf:typing}

\begin{figure}[t]
\begin{small}
\begin{mathpar}
\inferrule
    {\vdashl{C}{e}}
    {\vdashc{\apfenv, C}{\zldef{x}{e}}{\apfenv, C + \{x\}}}

\inferrule
    { }
    {\vdashc{\apfenv, C}{\zlnode{f}{x}{e}}{\apfenv + \{f \is \emptyset\}, C}}

\inferrule
    {\vdashc{C}{e}{\apftype}}
    {\vdashc{\apfenv, C}{\zlproba{f}{x}{e}}{\apfenv + \{f \is \apftype\}, C}}

\inferrule
    {\vdashc{\apfenv, C}{d_1}{\apfenv_1, C_1} \and \vdashc{\apfenv_1,C_1}{d_2}{\apfenv', C'}}
    {\vdashc{\apfenv, C}{d_1\ d_2}{\apfenv', C'}}

\\

\inferrule
	{ }
	{\vdashc{C}{c}{\emptyset}}

\inferrule
	{ }
	{\vdashc{C}{x}{\emptyset}}

\inferrule
	{\vdashc{C}{e_1}{\apftype_1}
     \and
     \vdashc{C}{e_2}{\apftype_2}}
	{\vdashc{C}{\zlpair{e_1}{e_2}}{\apftype_1 + \apftype_2}}

\inferrule
	{\vdashc{C}{e}{\apftype}}
	{\vdashc{C}{\zlop{e}}{\apftype}}

\inferrule
	{\vdashc{C}{e}{\apftype}}
	{\vdashc{C}{\zlsample{}{e}}{\apftype}}

\inferrule
	{\vdashc{C}{e}{\apftype}}
	{\vdashc{C}{\zlfactor{e}}{\apftype}}

\inferrule
	{ }
	{\vdashc{C}{\zllast{x}}{\emptyset}}

\inferrule
	{\vdashc{C}{e_1}{\apftype}}
	{\vdashc{C}{\zlpresent{e_1}{e_2}{e_3}}{\apftype}}

\inferrule
	{\vdashc{C}{e_2}{\apftype}}
	{\vdashc{C}{\zlreset{e_1}{e_2}}{\apftype}}

\inferrule
	{ }
	{\vdashc{C}{\zlapp{f_\theta}{e}}{\{\theta \is \zlfprior{f}\}}}

\inferrule
    {\vdashc{C}{e}{\apftype_e}\\
     \vdashleq{C}{E}{D}\\
     \vdashceq{C}{D}{E}{\apftype_E}}
    {\vdashc{C}{\zlwhere{e}{E}}{\apftype_e + \apftype_E}}

\inferrule
    {x \in D\\
     \vdashl{C}{e}}
    {\vdashceq{C}{D}{\zlisample{x}{}{e}}{\{x \is e\}}}

\inferrule
    {\vdashc{C}{e}{\apftype}}
    {\vdashceq{C}{D}{\zlinit{x}{e}}{\apftype}}

\inferrule
    {\vdashc{C}{e}{\apftype}}
    {\vdashceq{C}{D}{\zleq{x}{e}}{\apftype}}

\inferrule
    {\vdashceq{C}{D}{E_1}{\apftype_1}\\
     \vdashceq{C}{D}{E_2}{\apftype_2}}
    {\vdashceq{C}{D}{\zland{E_1}{E_2}}{\apftype_1 + \apftype_2}}
\end{mathpar}
\end{small}
\caption{Extract constant parameters and associated prior distributions.}
\label{fig:apf_sa-full}
\end{figure}

\begin{figure}[t]
\begin{small}
\begin{mathpar}
\inferrule
    { }
    {\vdashl{C}{c}}

\inferrule
    {x \in C}
    {\vdashl{C}{x}}

\inferrule
    {\vdashl{C}{e_1} \\ \vdashl{C}{e_2}}
    {\vdashl{C}{\zlpair{e_1}{e_2}}}

\inferrule
    {\vdashl{C}{e}}
    {\vdashl{C}{\zlop{e}}}

\inferrule
    {\vdashl{C + \dom{E}}{e}\\
     \vdashleq{C}{E}{\dom{E}}}
    {\vdashl{C}{\zlwhere{e}{E}}}

\inferrule
    { }
    {\vdashleq{C}{\zlinit{x}{e}}{\emptyset}}

\inferrule
    { }
    {\vdashleq{C}{\zleq{x}{\zllast{x}}}{\{x\}}}

\inferrule
    {\vdashl{C}{e}}
    {\vdashleq{C}{\zleq{x}{e}}{\{x\}}}

\inferrule
    {\notvdashl{C}{e}}
    {\vdashleq{C}{\zleq{x}{e}}{\emptyset}}

\inferrule
    {\vdashleq{C + C_2}{E_1}{C_1}\\
     \vdashleq{C + C_1}{E_2}{C_2}}
    {\vdashleq{C}{\zland{E_1}{E_2}}{C_1 + C_2}}
\end{mathpar}
\end{small}
\caption{Constant expressions and equations.}
\label{fig:apf_const-full}
\end{figure}

\Cref{fig:apf_const-full} presents the auxiliary type system identifying the constant expressions and extracting the names of the constant variables.
An expression is constant if it is a constant value~$c$, a variable referring to a constant stream~($x \in C$), or an expression whose all sub-expressions are constants.
The judgement $\vdashleq{C}{E}{C'}$ extracts all the name of all the streams of the set of equations~$E$ that are constant. So if $\vdashleq{C}{E}{\dom{E}}$, all the streams of~$E$ are constants.
An equation defines a constant streams if it is defined by a constant expression or by the equation~$\zleq{x}{\zllast{x}}$.
The full type system presented in \Cref{sec:apf:typing} is given in \Cref{fig:apf_sa-full}.

\subsection{Compilation}
\label{app:apf:compilation}

\begin{figure}
\begin{small}
\begin{array}[t]{@{}lcl@{}}
\compile{\apftype}{c} &=& c\\[0.2em]
\compile{\apftype}{x} &=& x\\[0.2em]
\compile{\apftype}{\zlpair{e_1}{e_2}} &=& \zlpair{\compile{\apftype}{e_1}}{\compile{\apftype}{e_2}}\\[0.2em]
\compile{\apftype}{\zlop{e}} &=& \zlop{\compile{\apftype}{e}}\\[0.2em]
\compile{\apftype}{\zllast{x}} &=& \zllast{x} \\[0.2em]
\compile{\apftype}{\zlpresent{e_1}{e_2}{e_3}} &=& \zlpresent{\compile{\apftype}{e_1}}{\compile{\apftype}{e_2}}{\compile{\apftype}{e_3}}\\[0.2em]
\compile{\apftype}{\zlreset{e_1}{e_2}} &=& \zlreset{\compile{\apftype}{e_1}}{\compile{\apftype}{e_2}}\\[0.2em]
\compile{\apftype}{\zlsample{}{e}} &=& \zlsample{}{\compile{\apftype}{e}}\\[0.2em]
\compile{\apftype}{\zlfactor{e}} &=& \zlfactor{\compile{\apftype}{e}}\\[0.2em]
\compile{\apftype}{\zlwhere{e}{E}} &=& \zlwhere{\compile{\apftype}{e}}{\compile{\apftype}{E}}\\[0.5em]

\compile{\apftype}{\zlinit{x}{e}} &=& 
    \left \{
    \begin{array}{@{}ll@{}}
        \emptyset & \text{if $x \in \dom{\apftype}$}\\
        \zlinit{x}{\compile{\apftype}{e}} & \text{otherwise}
    \end{array}
    \right .
\\[1em]
\compile{\apftype}{\zleq{x}{e}} &=& 
    \left \{
    \begin{array}{@{}ll@{}}
        \emptyset & \text{if $x \in \dom{\apftype}$}\\
        \zleq{x}{\compile{\apftype}{e}} & \text{otherwise}
    \end{array}
    \right .
\\[1em]

\compile{\apftype}{\zlapp{f_\theta}{e}} &=& 
  \left \{
  \begin{array}{@{}ll@{}}
    \zlapp{f}{\compile{\apftype}{e}} &\text{if $f$ is deterministic}\\
    \zlapp{\zlfmodel{f}}{\theta, \compile{\apftype}{e}} &\text{if $\theta \in \dom{\apftype$}}\\
    \zlapp{\zlfmodel{f}}{\theta, \compile{\apftype}{e}} 
    \ \kwf{where} &\text{otherwise} \\ 
    \quad 
    \kwf{rec} \  \zlinit{\theta}{\zlsample{}{\zlfprior{f}}} \\
    \quad 
    \kwf{and} \  \zleq{\theta}{\zllast{\theta}} \\
  \end{array}
  \right .\\[3em]

\compile{\apftype}{\zlinfer{\zlapp{f}{e}}} &=& \zlapfinfer{}{}{\zlfmodel{f}}{\zlfprior{f}}{\compile{\apftype}{e}}\\[1em]

\compile{\apftype}{\zland{E_1}{E_2}} &=& \zland{\compile{\apftype}{E_1}}{\compile{\apftype}{E_2}}\\[1em]

\compile{\apfenv}{\zldef{x}{e}} &=& \zldef{x}{e}\\[0.2em]
\compile{\apfenv}{\zlnode{f}{x}{e}} &=& \zlnode{f}{x}{\compile{\emptyset}{e}}\\[0.2em]
\compile{\apfenv}{\zlproba{f}{x}{e}} &=& 
  \begin{array}[t]{@{}l@{\quad}l}
  \zldef{\zlfprior{f}}{\im{\apftype}}
  & \text{with } \apftype = \apfenv(f)\\
  \zlproba{\zlfmodel{f}}{\zlpair{\dom{\apftype}}{x}}{\compile{\apftype}{e}}
  \end{array}
\end{array}
\end{small}
\caption{APF Compilation.}
\label{fig:apf:compilation-full}
\end{figure}  

The entire compilation function to transformation a ProbZelus model into a model compatible with \zl{APF.infer} is given in \Cref{fig:apf:compilation-full}.
Most cases simply call the compilation functions on all sub-expressions.
The interesting cases are presented in \Cref{sec:apf:compilation}.

\end{document}